%% file: main.tex
\documentclass[11pt,a4paper,reqno]{amsart}
\usepackage[utf8]{inputenc}
\usepackage{amsmath}
\usepackage{amssymb}
\usepackage{amsfonts}
\usepackage[toc,page]{appendix}
\usepackage{hyperref}
\usepackage{tikz}
\usepackage{color}
\usepackage{graphicx}
\usepackage{mathtools}
\usepackage[none]{hyphenat}
\usepackage{ytableau}
\usepackage{rotating}
\usepackage{pgfplots}
\usepackage{environ}
\usepackage{comment}

\usetikzlibrary{decorations.pathreplacing}
\newcommand{\qbinom}[2]{\genfrac{[}{]}{0pt}{}{#1}{#2}_{q}}

\newtheorem{theorem}{Theorem}
\newtheorem*{theorem*}{Theorem}
\theoremstyle{plain}

\newtheorem{corollary}{Corollary}

\newtheorem*{lemma*}{Lemma}

\newtheorem*{proposition*}{Proposition}
\theoremstyle{remark}

\newtheorem*{remark*}{Remark}
\theoremstyle{definition}
\newtheorem*{acknowledgements}{Acknowledgements}
\makeatletter
\newsavebox{\measure@tikzpicture}
\NewEnviron{scaletikzpicturetowidth}[1]{  \def\tikz@width{#1}  \def\tikzscale{1}\begin{lrbox}{\measure@tikzpicture}  \BODY
  \end{lrbox}  \pgfmathparse{#1/\wd\measure@tikzpicture}  \edef\tikzscale{\pgfmathresult}  \BODY
}
\makeatother

\allowdisplaybreaks
\begin{document}
\title[Matrix models for classical groups and Toeplitz$\pm $Hankel minors]{Matrix models for classical groups and Toeplitz$\pm $Hankel minors
with applications to Chern-Simons theory and fermionic models.}

\author{David Garc\'{\i}a-Garc\'{\i}a}
\address[DGG]{Grupo de F\'{i}sica Matem\'{a}tica, Departamento de Matem\'{a}tica, Faculdade de Ci\^{e}ncias, Universidade de Lisboa, Campo Grande, Edif\'{i}cio C6, 1749-016 Lisboa, Portugal.}
\email{dgarciagarcia@fc.ul.pt}
\author{Miguel Tierz}
\address[MT]{Departamento de Matem\'{a}tica, Faculdade de Ci\^{e}ncias, ISCTE - Instituto Universit\'{a}rio de Lisboa, Avenida das For\c{c}as Armadas, 1649-026 Lisboa, Portugal.}
\email{mtpaz@iscte-iul.pt}
\address[MT]{Grupo de F\'{i}sica Matem\'{a}tica, Departamento de Matem\'{a}tica, Faculdade de Ci\^{e}ncias, Universidade de Lisboa, Campo Grande, Edif\'{i}cio C6, 1749-016 Lisboa, Portugal.}
\email{tierz@fc.ul.pt}
\maketitle


\begin{abstract}
We study matrix integration over the classical Lie groups $U(N),Sp(2N),SO(2N)$ and $SO(2N+1)$, using symmetric function theory and the equivalent formulation in terms of determinants and minors of Toeplitz$\pm$Hankel matrices. We establish a number of factorizations and expansions for such integrals, also with insertions of irreducible characters. As a specific example, we compute both at finite and large $N$ the partition functions, Wilson loops and Hopf links of Chern-Simons theory on $S^{3}$ with the aforementioned symmetry groups. The identities found for the general models translate in this context to relations between observables of the theory. Finally, we use character expansions to evaluate
averages in random matrix ensembles of Chern-Simons type, describing the
spectra of solvable fermionic models with matrix degrees of freedom.
\end{abstract}


\section{Introduction}

There is a well known relation between matrix integrals over the classical Lie groups and the determinants of structured matrices, such as Toeplitz and Hankel matrices. These are matrices whose $(j,k)$-th coefficient depends only on $j-k$ or $j+k$, respectively, and are therefore constant along their diagonals or anti-diagonals. This connection is of importance to several areas of mathematics, such as random matrix theory and the theory of orthogonal polynomials \cite{BDJ,DIKrev}. At the same time, these two objects can be expressed in terms of symmetric functions, revealing further connections with enumerative combinatorics and representation theory \cite{BaikRains}.

In this work we study minors of the Toeplitz and Toeplitz$\pm$Hankel matrices involved in this relation. In addition to their own mathematical interest, one motivation for this arises from the fact that the minors of these matrices can be expressed as the ``twisted" integrals \cite{BumpDiaconis,GGT}
\begin{align}
    &\int_{G(N)}\chi_{G(N)}^{\lambda}(U^{-1})\chi_{G(N)}^{\mu}(U)f(U)dU, \label{first}
\intertext{where $dU$ denotes Haar measure on one of the classical Lie groups}
    &G(N) = U(N),Sp(2N),SO(2N),SO(2N+1), \nonumber
\end{align}
and the $\chi^{\lambda}_{G(N)}(U)$ are the characters associated to the irreducible representations of these groups.

Another motivation comes from the fact that \eqref{first} appears in the study of many contemporary physical theories and models. This is the case, for example, in gauge theories with a matrix model description, when one is interested in physical observables beyond the partition function and looks into non-local observables such as Wilson loops. The fact that tools germane to any of the above mentioned areas can be interchangeably applied to the analysis of such matrix models has not been fully exploited in the literature (see however \cite{Mironov:2017och}-\cite{vandeLeur:2016yuq}, for instance).

In particular, when $f$ is set to be Jacobi's third theta function in \eqref{first} we obtain the matrix model of Chern-Simons theory on $S^{3}$ with symmetry group $G(N)$. After a matrix model description was obtained for Chern-Simons theory on manifolds such as $S^{3}$ or lens spaces \cite{Marino:2002fk}, the solvability of the theory has been well known, and a number of equivalent representations have been obtained \cite{Tierz:2002jj}-\cite{Romo:2011qp}. However, while both the partition function and the observables of the unitary theory are known and have been studied in detail, much less attention has been devoted to the symplectic or orthogonal theories, where only the partition function in the large $N$ regime has been obtained \cite{Sinha:2000ap}.

It is worth mentioning that the determinants of Toeplitz$\pm$Hankel matrices have many applications in statistical mechanics problems and describe several physical properties of a number of strongly correlated systems, starting with their appearance in the Ising model \cite{DIKrev}. In such applications, the Toeplitz$\pm$Hankel case corresponds to open boundary conditions, whereas the Toeplitz determinants correspond to periodic boundary conditions \cite{Viti:2016hty}-\cite{Perez-Garcia:2013lba}. The study of minors is less developed but, in the spin chain context, they naturally appear in the same fashion as the determinants, allowing the treatment of quantum amplitudes involving multiple domain wall configurations \cite{Perez-Garcia:2013lba}-\cite{SantilliTierzSpin}, whereas a single domain is given directly by the determinant \cite{Viti:2016hty}.

We pursue two main goals with the present work:

\begin{enumerate}
    \item First, we use the formulation of matrix integrals as determinants of Toeplitz$\pm$Hankel matrices and exploit their relations with symmetric functions to establish a number of identities between these objects. In particular, we show that there is a factorization property for matrix integration over $U(2N-1)$ and $U(2N)$ in terms of matrix integration over symplectic and orthogonal groups. We also see how group integrals of polynomial functions can be expressed as the specialization of a single symmetric function. We then show that any $G(N)$ matrix integration can be written as a finite sum of twisted $U(N)$ integrals, or, equivalently, that determinants of Toeplitz$\pm$Hankel matrices can be written as finite sums of minors of a Toeplitz matrix. Finally, we express matrix integrals over $G(N)$ as Schur function series, obtaining in particular that the normalized averages of two characters over a $G(N)$ ensemble have the same behavior for large $N$. Other relations between unitary, symplectic and orthogonal matrix models have been investigated in \cite{ForrRains}, and recent related generalizations of classical results for Toeplitz matrices to the Toeplitz$\pm$Hankel setting include \cite{DIK}-\cite{Betea}, for instance.
    
    \item We then study in detail the case where $f$ is a theta function. The reason is because the corresponding determinants and minors can be computed exactly for finite matrix size $N$ and, in addition, the results have a topological interpretation, since the expressions obtained can be written in terms of the modular $S$ and $T$ matrices. Quantum invariants of manifolds and links can also be approached with skein theory and quantum groups \cite{Turaev:1994xb,Morton} and in fact the same determinant representation as in the unitary model arises when studying the skein module of the annulus \cite{MortonLukac}.
    
    We remark that the symmetric function approach allows a unified treatment for all of the groups $G(N)$, as well as generalizations of some properties usually attributed only to unitary ensembles, such as preservation of Schur polynomials \cite{Mironov:2017och,Morozov:2018eiq} or Giambelli compatibility \cite{GiamBor}. Note also that the previously obtained results have now an interpretation in terms of Chern-Simons observables. For example, we show that $G(N)$ Chern-Simons partition functions can be expressed as sums of unnormalized Hopf links ($S$ matrices) of the $U(N)$ theory.

\end{enumerate}

These methods and results can also be quickly adapted to study some fermionic exactly solvable models, that have recently been obtained in the study of fermionic quantum models with matrix degrees of freedom \cite{Anninos:2016klf}-\cite{Klebanov:2018nfp}. Some of these models appear as simpler cases of tensor quantum mechanical models, of much interest nowadays \cite{Klebanov:2018nfp}. Following the example of the matrix model representation obtained in \cite{Anninos:2016klf} and solved in \cite{Tierz:2017nvl} for a system of $U(N)\times U(L)$ fermions with a finite Hilbert space, we consider analogous matrix model expressions coming from matrix integration over other Lie groups. In particular, we study partition functions of such models, defined as averages of characteristic polynomial type in $G(N)$ Chern-Simons matrix models, and obtain the distinctive oscillator like and highly degenerated spectrum of the models \cite{Tierz:2017nvl,Klebanov:2018nfp}.

The paper is organized as follows: In Section 2, after introducing the required definitions and the equivalence between integration over the classical groups $G(N)$ and determinants of Toeplitz$\pm$Hankel matrices, we establish several general relations that hold among the integrals \eqref{first} and their symmetric function counterparts.

Throughout the rest of the paper we turn to the Chern-Simons model. In Section 3, we evaluate the corresponding determinants and obtain explicit expressions for the $G(N)$ Chern-Simons partition functions, for both finite and large $N$. In Section 4, we continue and evaluate the Wilson loops and Hopf links of the theory, which correspond to the minors of the underlying matrices.

In the last Section, we study partition functions of fermionic matrix models as averages of characteristic polynomials in the $G(N)$ Chern-Simons matrix models, which we show can be computed with character expansions. Through the explicit evaluation of partition functions, for both massive and massless cases, we characterize the corresponding spectra and relate it to the spectra of fermionic models with matrix degrees of freedom. We also obtain large $N$ expressions for these models, using character expansion and Fisher-Hartwig asymptotics \cite{DIKrev}.


\section{Group integrals, Toeplitz$\pm$Hankel matrices and characters of the
classical groups}

Let $U$ be a random matrix distributed according to normalized Haar measure on one of the classical groups $G(N)=U(N),Sp(2N),SO(2N),SO(2N+1)$. Matrices in $SO(2N+1)$ have $1$ as a trivial eigenvalue, and the remaining eigenvalues of matrices in $Sp(2N),SO(2N)$ and $SO(2N+1)$ are complex numbers of modulus one that come in complex conjugate pairs. We say that the eigenvalues in the lower half plane are also trivial.

Given an integrable function on the unit circle $f$, we define\footnote{We choose to employ this abuse of notation in favor of a simpler writing.}
\begin{equation}
    f(U)=\prod_{k=1}^{N}f(e^{i\theta_{k}})f(e^{-i\theta_{k}}),  \label{fprod}
\end{equation}
for any matrix $U$ belonging to one of the groups $G(N)=U(N),Sp(2N),SO(2N),SO(2N+1)$, where $e^{i\theta_{1}},\dots,e^{i\theta_{N}} $ are the nontrivial eigenvalues of $U$ (which, in the unitary case, coincide with the full set of eigenvalues of $U$). If we denote by $\int_{G(N)}f(U)dU$ the integral of this function over one of the groups $G(N)$ with respect to Haar measure, Weyl's integral formula \cite{Weyl} reads
\begin{align*}
    &\int_{U(N)}f(U)dU = \frac{1}{N!}\int_{[0,2\pi]^{N}}\prod_{j<k}\left|e^{i\theta_{j}}-e^{i\theta_{k}}\right|^{2}\prod_{k=1}^{N}f(e^{i\theta_{k}})f(e^{-i\theta_{k}})\frac{d\theta_{k}}{2\pi}, \\
    &\int_{Sp(2N)}f(U)dU = 2^{N^{2}+N}\frac{1}{N!}\int_{[0,\pi]^{N}}\prod_{j<k}\left(\cos{\theta_{j}}-\cos{\theta_{k}}\right)^{2}\prod_{k=1}^{N}\sin^{2}{\theta_{k}}\prod_{k=1}^{N}f(e^{i\theta_{k}})f(e^{-i\theta_{k}})\frac{d\theta_{k}}{2\pi}, \\
    &\int_{SO(2N)}f(U)dU = 2^{N^{2}-N+1}\frac{1}{N!}\int_{[0,\pi]^{N}}\prod_{j<k}\left(\cos{\theta_{j}}-\cos{\theta_{k}}\right)^{2}\prod_{k=1}^{N}f(e^{i\theta_{k}})f(e^{-i\theta_{k}})\frac{d\theta_{k}}{2\pi}, \\
    &\int_{SO(2N+1)}f(U)dU = 2^{N^{2}+N}\frac{1}{N!}\int_{[0,\pi]^{N}}\prod_{j<k}\left(\cos{\theta_{k}}-\cos{\theta_{j}}\right)^{2}\prod_{k=1}^{N}\sin^{2}{\frac{\theta_{k}}{2}}\prod_{k=1}^{N}f(e^{i\theta_{k}})f(e^{-i\theta_{k}})\frac{d\theta_{k}}{2\pi}.
\end{align*}
Using the invariance of the integrands in the right-hand sides of the last three equations above upon the transformation $\theta\mapsto-\theta$ and substituting the sines and cosines above by their expressions in terms of the $e^{i\theta_{j}}$ we see that the following compact expression is available for the group integrals of the function \eqref{fprod}
\begin{equation}
    \int_{G(N)}f(U)dU = C_{G(N)}\frac{1}{N!}\int_{[0,2\pi]^{N}}\det(M_{G(N)}(e^{-i\theta}))\det(M_{G(N)}(e^{i\theta}))\prod_{k=1}^{N}f(e^{i\theta_{k}})f(e^{-i\theta_{k}})\frac{d\theta_{k}}{2\pi},  \label{intGN}
\end{equation}
where the constants $C_{G(N)}$ are 
\begin{equation*}
    C_{U(N)}=1,\qquad C_{Sp(2N)}=\frac{1}{2^{N}}=C_{SO(2N+1)},\qquad C_{SO(2N)}=\frac{1}{2^{N+1}}
\end{equation*}
and $M_{G(N)}(e^{i\theta})$ is the matrix appearing in Weyl's denominator formula\footnote{While the expression \eqref{intGN} follows from the trigonometric expressions for Weyl's integral formula and the determinants \eqref{detun}-\eqref{detone}, as outlined above, a compact formula such as \eqref{intGN} is present already in the original derivation of Weyl \cite{Weyl}.} for the root system associated to each of the groups $G(N)$. See \eqref{detun}-\eqref{detone} for explicit expressions of these matrices and their determinants. This identity makes it possible to obtain equivalent determinantal expressions by means of the following classical identity due to Andrei\'ef \cite{Andreief}.

\begin{lemma*}
Let $g_{1},\dots,g_{N}$ and $h_{1},\dots,h_{N}$ be functions on a measure space $(X,\sigma)$. Then,
\begin{equation*}
    \frac{1}{N!}\int_{X^{N}}\det{(g_{j}(x_{k}))}_{j,k=1}^{N}\det{(h_{j}(x_{k}))}_{j,k=1}^{N}\prod_{k=1}^{N}d\sigma(x_{k}) = \det{\left(\int_{X}g_{j}(x)h_{k}(x)d\sigma(x)\right)_{j,k=1}^{N}},
\end{equation*}
as long as both the left- and right-hand sides above are well-defined.
\end{lemma*}

The integrals \eqref{intGN} can be written in the form above, choosing $d\sigma(e^{i\theta}) = f(e^{i\theta})f(e^{-i\theta})d\theta/2\pi$ for $\theta\in[0,2\pi)$ as measure and suitable functions $g_{j}$ and $h_{j}$ for each of the groups $G(N)$ (for instance, $h_{j}(e^{i\theta})=g_{j}(e^{-i\theta})=e^{i(N-j)\theta}$ for $j=1,\dots,N$ for $U(N)$, see equations \eqref{detun}-\eqref{detone}). A direct application of Andrei\'ef's identity in \eqref{intGN} then yields
\begin{align}
    &\int_{U(N)}f(U)dU = \det{\left(d_{j-k}\right)_{j,k=1}^{N}},  \label{intUN} \\
    &\int_{Sp(2N)}f(U)dU = \det{\left( d_{j-k}-d_{j+k} \right)_{j,k=1}^{N}}, \label{intSp2N} \\
    &\int_{SO(2N)}f(U)dU = \frac{1}{2}\det{\left(d_{j-k}+d_{j+k-2}\right)_{j,k=1}^{N}},  \label{intO2N} \\
    &\int_{SO(2N+1)}f(U)dU = \det{\left(d_{j-k}-d_{j+k-1}\right)_{j,k=1}^{N}}, \label{intO2N+1}
\end{align}
where $d_{k}$ denotes the Fourier coefficient 
\begin{equation}  \label{fourier}
    d_{k} = \frac{1}{2\pi}\int_{0}^{2\pi}e^{ik\theta}f(e^{i\theta})f(e^{-i\theta})d\theta
\end{equation}
for each $k\in\mathbb{Z}$ (note that $d_{k}=d_{-k}$ for all $k$). Expressions for group integrals as determinants of Toeplitz$\pm$Hankel matrices have been obtained previously, see for instance \cite{BaikRains}. Besides their own intrinsic interest, matrix integrals over the groups $G(N)$ enjoy connections with combinatorics \cite{BaikRains}, number theory \cite{KeatingSnaith2} and integrable systems \cite{AvM}, among many other topics.

Given $U\in G(N)$, we also define
\begin{equation*}
    f(-U) = \prod_{k=1}^{N}f(-e^{i\theta_{k}})f(-e^{-i\theta_{k}}),
\end{equation*}
where the $e^{i\theta_{k}}$ are the nontrivial eigenvalues of $U$. While $\int_{G(N)}f(U)dU = \int_{G(N)}f(-U)dU$ for $G(N)=U(N),Sp(2N),SO(2N)$ (as follows from the above determinantal expressions, for instance), in the odd orthogonal case we have
\begin{equation}
    \int_{SO(2N+1)}f(-U)dU = \det{\left(d_{j-k}+d_{j+k-1}\right)_{j,k=1}^{N}}. \label{intO2N+1-}
\end{equation}

The irreducible representations of the groups $G(N)$ are indexed by partitions \cite{FultonHarris,KoikeTerada} (see appendix A for the definition and some basic facts about partitions). We will write $\chi_{G(N)}^{\lambda}$ to denote the character of the group $G(N)$ indexed by the partition $\lambda$. These can be expressed as the 
quotient of a minor of the corresponding matrix $M_{G(N)}(e^{i\theta})$, obtained by striking some of its columns, over the determinant of the matrix itself, see \eqref{unchar}-\eqref{oevenchar}. Hence, the insertion of one or two characters of the group $G(N)$ in the integrand in \eqref{intGN} cancels one or two of the determinants. Therefore, using Andrei\'{e}f's identity again on the resulting integral we obtain the following.

\begin{theorem} \label{th.minorsTH}
Let $\lambda$ and $\mu$ be two partitions of lengths $l(\lambda),l(\mu)\leq N$, and define the ``reversed" arrays $\lambda^{r}$ and $\mu^{r}$ as
\begin{equation*}
    \lambda^{r} = (\lambda_{N-j+1})_{j} = (\lambda_{N},\lambda_{N-1},\dots,\lambda_{2},\lambda_{1}),\qquad \mu^{r} = (\mu_{N-j+1})_{j} = (\mu_{N},\dots,\mu_{1}).
\end{equation*}
We then have
\begin{align*}
    &\int_{U(N)}\chi^{\lambda}_{U(N)}(U^{-1})\chi^{\mu}_{U(N)}(U)f(U)dU = \det {\left(d_{j-\lambda_{j}-k+\mu_{k}}\right) _{j,k=1}^{N}}=\det{\left(d_{j+\lambda_{j}^{r}-k-\mu_{k}^{r}}\right)_{j,k=1}^{N}}, \\ 
    &\int_{Sp(2N)}\chi^{\lambda}_{Sp(2N)}(U)\chi^{\mu}_{Sp(N)}(U)f(U)dU = \det{\left(d_{j+\lambda_{j}^{r}-k-\mu_{k}^{r}}-d_{j+\lambda_{j}^{r}+k+\mu_{k}^{r}}\right)_{j,k=1}^{N}}, \\
    &\int_{SO(2N)}\chi^{\lambda}_{SO(2N)}(U)\chi^{\mu}_{SO(2N)}(U)f(U)dU = \frac{1}{2}\det{\left(d_{j+\lambda_{j}^{r}-k-\mu_{k}^{r}}+d_{j+\lambda_{j}^{r}+k+\mu_{k}^{r}-2}\right)_{j,k=1}^{N}}, \\
    &\int_{SO(2N+1)}\chi^{\lambda}_{SO(2N+1)}(U)\chi^{\mu}_{SO(2N+1)}(U)f(U)dU = \det{\left(d_{j+\lambda_{j}^{r}-k-\mu_{k}^{r}}-d_{j+\lambda_{j}^{r}+k+\mu_{k}^{r}-1}\right)_{j,k=1}^{N}},
\end{align*}
where the $d_{k}$ are given by \eqref{fourier}.
\end{theorem}

We have used above the fact that $\chi_{G(N)}^{\lambda }(U)=\chi_{G(N)}^{\lambda}(U^{-1})$ for $G(N)=Sp(2N),SO(2N),SO(2N+1)$.

The resulting determinants are now minors of the Toeplitz and Toeplitz$\pm $Hankel matrices appearing in the right hand sides of formulas \eqref{intUN}-\eqref{intO2N+1}, obtained by striking some of their rows and columns. This was already noted for the $U(N)$ case in \cite{BumpDiaconis}. Moreover, the precise striking of rows and columns performed on the underlying matrix only depends on the partitions $\lambda$ and $\mu$, and is the same for any of the matrices \eqref{intUN}-\eqref{intO2N+1}. These strikings can be read off from the partitions, see \cite{BumpDiaconis},\cite{GGT} for an explicit algorithm.

Let us show some examples of how these determinant and minor expressions can be exploited to obtain some known and new results.

\subsection{Factorizations}

\begin{theorem} \label{thdets}
We have 
\begin{align*}
    \int_{U(2N-1)}f(U)dU& =\int_{Sp(2N-2)}f(U)dU\int_{SO(2N)}f(U)dU \\
    & =\frac{1}{2}\int_{SO(2N-1)}f(U)dU\int_{SO(2N+1)}f(-U)dU+\frac{1}{2}\int_{SO(2N+1)}f(U)\int_{SO(2N-1)}f(-U)dU, \\
    \int_{U(2N)}f(U)dU& =\int_{SO(2N+1)}f(U)dU\int_{SO(2N+1)}f(-U)dU \\
    & =\frac{1}{2}\int_{Sp(2N)}f(U)dU\int_{SO(2N)}f(U)dU+\frac{1}{2}\int_{Sp(2N-2)}f(U)\int_{SO(2N+2)}f(U)dU.
\end{align*}
\end{theorem}

\begin{proof}
The theorem follows immediately after expressing the above integrals as the Toeplitz and Toeplitz$\pm$Hankel determinants \eqref{intUN}-\eqref{intO2N+1},\eqref{intO2N+1-} and noticing that these determinants satisfy the corresponding identities, see e.g. \cite{VeinDale}.
\end{proof}

The characters $\chi^{\lambda}_{G(N)}$ can be lifted to the so called ``universal characters" in the ring of symmetric functions in countably many variables \cite{KoikeTerada}. In this fashion, the lifting of the characters of $U(N)$, $Sp(2N)$, $SO(2N)$ and $SO(2N+1)$ gives rise to the Schur $s_{\lambda}$, symplectic Schur $sp_{\lambda}$, even orthogonal Schur $o^{even}_{\lambda}$ and odd orthogonal Schur $o^{odd}_{\lambda}$ functions, respectively. See \eqref{JTschur}-\eqref{JToortd} for explicit expressions of these functions. When the length of the partition $\lambda$ is less than or equal to the number of nontrivial eigenvalues of a matrix $U$, these functions coincide with the irreducible characters of the corresponding group, after specializing the corresponding variables back to the nontrivial eigenvalues $z_{j}$ of $U$. For instance, we have $\chi^{\lambda}_{Sp(2N)}(U) = sp_{\lambda}(z_{1},\dots,z_{N})$ for any partition satisfying $l(\lambda)\leq N$. Note that this condition is necessary in order for the characters $\chi^{\lambda}_{G(N)}(U)$ to be defined, while the symmetric functions \eqref{JTschur}-\eqref{JToortd} need not satisfy such restriction, and are defined for more general partitions. See appendix A and \cite{KoikeTerada} for details on this, as well as some properties fulfilled by these functions. The close relation between these two families has further consequences, as we will see throughout this section.

Given a partition $\lambda$ satisfying $l(\lambda)\leq N$ and $\lambda_{1}\leq K$ (that is, $\lambda\subset(K^{N})$), we define a new partition by
\begin{equation} \label{rotpart}
    L_{K,N}(\lambda) = (K-\lambda_{N},\dots,K-\lambda_{1})=(K^{N})-\lambda^{r},
\end{equation}
where $\lambda^{r}$ denotes the ``reversed" array $(\lambda_{N},\dots,\lambda_{1})$. That is, $L_{K,N}(\lambda)$ is the partition that results from rotating $180$º the complement of $\lambda$ in the rectangular shape $(K^{N})$. With the aid of this, we can state the next result.

\begin{theorem} \label{thsingle}
Let $x=(x_{1},\dots,x_{K})$ be some variables, and let $\lambda$ be a partition satisfying $l(\lambda)\leq N$ and $\lambda_{1}\leq K$. We have
\begin{align}
    &\int_{Sp(2N)}\chi_{Sp(2N)}^{\lambda}(U)\prod_{j=1}^{K}(1+x_{j}U)dU = \left( \prod_{j=1}^{K}x_{j}^{N}\right) sp_{L_{N,K}(\lambda')}(x_{1},\dots,x_{K}) \label{spsingle} \\
    &\int_{SO(2N)}\chi_{SO(2N)}^{\lambda}(U)\prod_{j=1}^{K}(1+x_{j}U)dU = \left( \prod_{j=1}^{K}x_{j}^{N} \right) o^{even}_{L_{N,K}(\lambda')}(x_{1},\dots,x_{K}) \label{oesingle} \\
    &\int_{SO(2N+1)}\chi_{SO(2N+1)}^{\lambda}(U)\prod_{j=1}^{K}(1+x_{j}U)dU = (-1)^{|\lambda|+KN}\left( \prod_{j=1}^{K}x_{j}^{N} \right) o^{odd}_{L_{N,K}(\lambda')}(-x_{1},\dots,-x_{K}), \label{oosingle}
\end{align}
where $L_{N,K}(\lambda')$ is the partition given by \eqref{rotpart}.
\end{theorem}

\begin{proof}
Let us proceed with the symplectic case. We start from the case $\mu=\varnothing$ of the symplectic integral in theorem \ref{th.minorsTH}. Using the well known fact that
\begin{equation*}
    \prod_{j=1}^{K}(1+x_{j}z) = \sum_{k=0}^{K}e_{k}(x)z^{k},
\end{equation*}
where the $e_{k}(x)$ are the elementary symmetric polynomials \eqref{ek} on the variables $x_{1},\dots,x_{K}$, we see that the $k$-th Fourier coefficient \eqref{fourier} for this choice of function is
\begin{equation*}
    d_{k} = \left(\prod_{j=1}^{K}x_{j}\right)e_{K+k}(x,x^{-1}).
\end{equation*}
We thus have
\begin{align} 
    \int_{Sp(2N)}\chi_{Sp(2N)}^{\lambda}(U)&\prod_{j=1}^{K}(1+x_{j}U)dU \\
    &=\det{\left( \prod_{j=1}^{K}x_{j}\left( e_{K+j+\lambda^{r}_{j}-k}(x,x^{-1})-e_{K+j+\lambda^{r}_{j}+k}(x,x^{-1})\right)\right)_{j,k=1}^{N}}, \label{auxthsingle}
\end{align}
where we have denoted $x^{-1}=(x_{1}^{-1},\dots,x_{K}^{-1})$. Now, since
\begin{equation}
    e_{j}(x_{1},\dots,x_{K},x_{1}^{-1},\dots,x_{K}^{-1}) = e_{2K-j}(x_{1},\dots,x_{K},x_{1}^{-1},\dots,x_{K}^{-1}), \label{elemid}
\end{equation}
as follows from \eqref{ek}, we see that the determinant in \eqref{auxthsingle} can also be expressed as
\begin{align*}
    \det{\left( \prod_{j=1}^{K}x_{j}\left( e_{K-\lambda_{N+1-j}-j+k}(x,x^{-1})-e_{K-\lambda_{N+1-j}-j-k}(x,x^{-1})\right)\right)_{j,k=1}^{N}},
\end{align*}
which, due to the Jacobi-Trudi identity \eqref{JTsympd}, coincides with the right hand side of \eqref{spsingle}.

Identity \eqref{oesingle} follows analogously. Let us turn however, to identity \eqref{oosingle}, as it requires some more computation. As in the symplectic case, using the Jacobi-Trudi identity \eqref{JToortd}, the fact that $e_{k}(x,1)=e_{k}(x)+e_{k-1}(x)$, and identity \eqref{elemid} we obtain
\begin{align*}
    \left(\prod_{j=1}^{K}x_{j}^{N}\right) &o_{L_{N,K}(\lambda')}^{odd}(-x) \\
    = \frac{1}{2} \det & \left(\prod_{j=1}^{K}x_{j}\left(e_{K-\lambda_{j}^{r}-j+k}(-x,-x^{-1},1)+e_{K-\lambda_{j}^{r}-j-k+2}(-x,-x^{-1},1)\right)\right)_{j,k=1}^{N} \\
    = \frac{1}{2} \det & \left( \prod_{j=1}^{N}x_{j}\left( e_{K-\lambda_{j}^{r}-j+k}(-x,-x^{-1})+e_{K-\lambda_{j}^{r}-j+k-1}(-x,-x^{-1}) \right. \right. \\
    & \left. \left. \hphantom{{}\prod_{j=1}^{K}x_{j}++} +e_{K-\lambda_{j}^{r}-j-k+2}(-x,-x^{-1})+e_{K-\lambda_{j}^{r}-j-k+1}(-x,-x^{-1}) \right) \right)_{j,k=1}^{N} \\
    = \frac{1}{2} \det & \left( \prod_{j=1}^{N}x_{j} \left( e_{K+j+\lambda_{j}^{r}-k}(-x,-x^{-1})+e_{K+j+\lambda_{j}^{r}-k+1}(-x,-x^{-1}) \right. \right. \\
    & \left. \left. \hphantom{{}\prod_{j=1}^{K}x_{j}++} +e_{K+j+\lambda_{j}^{r}+k-2}(-x,-x^{-1})+e_{K+j+\lambda_{j}^{r}+k-1}(-x,-x^{-1}) \right) \right)_{j,k=1}^{N}.
\end{align*}
Adding $(-1)^{j+k}$ times the $k$-th column of the last matrix above, for each $k=1,...,j-1$, to the $j$-th column, for each $j=2,...,N$, we obtain
\begin{equation*}
    \left(\prod_{j=1}^{K}x_{j}^{N}\right)o_{L_{N,K}(\lambda')}^{odd}(-x) = \det{\left(\prod_{j=1}^{K}x_{j}\left( e_{K+j+\lambda_{j}^{r}-k}(-x,-x^{-1})+e_{K+\lambda_{j}^{r}+j+k-1}(-x,-x^{-1})\right)\right)_{j,k=1}^{N}}.
\end{equation*}
Using the case $\mu=\varnothing$ of the odd orthogonal integral of theorem \ref{th.minorsTH} and extracting the minus sign from the elementary symmetric polynomials in the last determinant above we arrive at \eqref{oosingle}.
\end{proof}

In particular, theorem \ref{thsingle} implies that the determinants of the corresponding Toeplitz$\pm$Hankel matrices in the left hand sides of the theorem can be expressed as the specialization of a single character associated to the irreducible representation of the corresponding group, indexed by a rectangular partition. This was first observed in \cite{Conreyetal} and has been generalized to integrals over other ensembles, see for instance \cite{MatsumotoRect,Matsumoto}. Combining this fact with theorem \ref{thdets} we obtain the following result.

\begin{corollary} \label{cofacts}
The following relations hold between the symmetric functions associated to the characters of the groups $G(N)$
\begin{align*}
    &s_{((2N-1)^{K})}(x_{1},\dots,x_{K},x_{1}^{-1},\dots,x_{K}^{-1}) = sp_{((N-1)^{K})}(x_{1},\dots,x_{K})o^{even}_{(N^{K})}(x_{1},\dots,x_{K}) \\
    &= \frac{(-1)^{NK}}{2}o^{odd}_{((N-1)^{K})}(x_{1},\dots,x_{K})o^{odd}_{(N^{K})}(-x_{1},\dots,-x_{K}) \\
    &+\frac{(-1)^{NK}}{2}o^{odd}_{(N^{K})}(x_{1},\dots,x_{K})o^{odd}_{((N-1)^{K})}(-x_{1},\dots,-x_{K}), \\
    &s_{((2N)^{K})}(x_{1},\dots,x_{K},x_{1}^{-1},\dots,x_{K}^{-1}) = (-1)^{NK}o^{odd}_{(N^{K})}(x_{1},\dots,x_{K})o^{odd}_{(N^{K})}(-x_{1},\dots,-x_{K}) \\
    &= \frac{1}{2}sp_{(N^{K})}(x_{1},\dots,x_{K})o^{even}_{(N^{K})}(x_{1},\dots,x_{K}) +\frac{1}{2}sp_{((N-1)^{K})}(x_{1},\dots,x_{K})o^{even}_{((N+1)^{K})}(x_{1},\dots,x_{K}).
\end{align*}
\end{corollary}

The first and third identities in the corollary appeared before in \cite{CiucuKratt}. There exist also identities expressing the sum of two Schur polynomials indexed by partitions of rectangular shapes in terms of orthogonal and symplectic Schur functions, as well as some other generalizations of these identities, see \cite{CiucuKratt,ABfact,AFfact}, but the second and fourth identities are new to our knowledge.

\subsection{Expansions in terms of Toeplitz minors} \label{s.minexp}

Let us recall the Frobenius notation for partitions before stating the next result. Let $\nu$ be a partition; we denote $\nu =(a_{1},\dots,a_{p}|b_{1},\dots ,b_{p})$, for some nonnegative integers $a_{1}>\dots >a_{p}$ and $b_{1}>\dots >b_{p}$, if there are $p$ boxes on the main diagonal of the Young diagram of $\nu $, with the $k$-th box having $a_{k}$ boxes immediately to the right and $b_{k}$ boxes immediately below. We denote by $p(\nu)$ the number of boxes on the main diagonal of the diagram of a partition $\nu$. With this notation, we can introduce the sets $R(N),S(N)$ and $T(N)$ of partitions of shapes $(a_{1}+1,\dots ,a_{p}+1|a_{1},\dots ,a_{p})$, $(a_{1},\dots ,a_{p}|a_{1},\dots ,a_{p})$ and $(a_{1}-1,\dots ,a_{p}-1|a_{1},\dots ,a_{p})$ respectively in Frobenius notation, with $a_{1}\leq N-1$. For instance, the set $R(3)$ consists of the partitions \ytableausetup{boxsize=0.2cm} 
\begin{equation*}
    \left\{ \varnothing ,\ydiagram{2},\ydiagram{3,1},\ydiagram{3,3},\ydiagram{4,1,1},\ydiagram{4,3,1},\ydiagram{4,4,2},\ydiagram{4,4,4}\right\} ,
\end{equation*}
the set $S(3)$ is the set of self-conjugate partitions of length at most $3$ and the set $T(3)$ is obtained as the set of partitions conjugated to those of $R(2)$. Note that there are exactly $2^{N}$ partitions in each of the sets $R(N)$ and $S(N)$, and $2^{N-1}$ in the set $T(N)$, all of them of length less than or equal to $N$.

\begin{theorem}
\label{th.hopflinks} The integrals \eqref{intGN} verify 
\begin{align*}
    &\int_{Sp(2N)}f(U)dU = \frac{1}{2^{N}}\sum_{\rho_{1},\rho_{2}\in R(N)}(-1)^{(|\rho_{1}|+|\rho_{2}|)/2}\int_{U(N)}\chi_{U(N)}^{\rho_{1}}(U^{-1})\chi_{U(N)}^{\rho_{2}}(U)f(U)dU, \\
    &\int_{SO(2N)}f(U)dU = \frac{1}{2^{N-1}}\sum_{\tau_{1},\tau_{2}\in T(N)}(-1)^{(|\tau_{1}|+|\tau_{2}|)/2}\int_{U(N)}\chi_{U(N)}^{\tau_{1}}(U^{-1})\chi_{U(N)}^{\tau_{2}}(U)f(U)dU, \\
    &\int_{SO(2N+1)}f(U)dU = \frac{1}{2^{N}}\sum_{\sigma_{1},\sigma_{2}\in S(N)}(-1)^{(|\sigma_{1}|+|\sigma_{2}|+p(\sigma_{1})+p(\sigma_{2}))/2}\int_{U(N)}\chi_{U(N)}^{\sigma_{1}}(U^{-1})\chi_{U(N)}^{\sigma_{2}}(U)f(U)dU
\end{align*}
\end{theorem}

That is, the integral of a function over one of the groups $G(N)$ can be expressed as a certain sum of integrals of the same function over $U(N)$ with Schur polynomials on the integrand. Note that the integrals in the right hand sides above are symmetric upon exchange of the partitions indexing the Schur polynomials. This\footnote{Together with further symmetries of the integral; for instance, $\int_{U(N)}s_{(a^{N})}(U^{-1})s_{(a^{N})}(U)f(U)dU = \int_{U(N)}f(U)dU$ for every $a>0$.} implies that there are $2^{2N-1}$ different terms in each of the sums.
\begin{proof}
The main idea is that the determinants $\det{M_{G(N)}}(z)$, for $G(N)=Sp(2N),SO(2N),SO(2N+1)$, when seen as symmetric functions, contain as a factor the determinant $\det{M_{U(N)}}(z)$ (see formulas \eqref{detun}-\eqref{detone}). Hence, as a consequence of the definition \eqref{fprod}, one can see the integrals over the groups $G(N)$ as integrals over $U(N)$ with an additional term in the integrand. Moreover, these additional terms can be expressed as Schur functions series as follows \cite{MacDonald} 
\begin{align*}
    \frac{\det{M_{Sp(2N)}(z)}}{\det{M_{U(N)}(z)}} &= \prod_{j=1}^{N}z_{j}^{-N}\prod_{j<k}(1-z_{j}z_{k})\prod_{j=1}^{N}(1-z_{j}^{2}) = \prod_{j=1}^{N}z_{j}^{-N}\sum_{\rho\in R(N)}(-1)^{|\rho|/2}s_{\rho}(z_{1},\dots,z_{N}), \\
    \frac{\det{M_{SO(2N)}(z)}}{\det{M_{U(N)}(z)}} &= 2\prod_{j=1}^{N}z_{j}^{-N+1}\prod_{j<k}(1-z_{j}z_{k})\prod_{j=1}^{N} = 2\prod_{j=1}^{N}z_{j}^{-N+1}\sum_{\tau\in T(N)}(-1)^{|\tau|/2}s_{\tau}(z_{1},\dots,z_{N}), \\
    \frac{\det{M_{SO(2N+1)}(z)}}{\det{M_{U(N)}(z)}} &= \prod_{j=1}^{N}z_{j}^{-N+1/2}\prod_{j<k}(1-z_{j}z_{k})\prod_{j=1}^{N}(1-z_{j}) \\
    &= \prod_{j=1}^{N}z_{j}^{-N+1/2}\sum_{\sigma\in S(N)}(-1)^{(|\sigma|+p(\sigma))/2}s_{\sigma}(z_{1},\dots,z_{N}).
\end{align*}
Substituting these formulas into \eqref{intGN}, for each of the groups $G(N)=Sp(2N),SO(2N),SO(2N+1)$, one obtains the desired result.
\end{proof}

According to identities \eqref{intUN}-\eqref{intO2N+1}, the integrals and twisted integrals over the groups $G(N)$ can be expressed as determinants and minors, respectively, of certain Toeplitz$\pm$Hankel matrices. Therefore, theorem \ref{th.hopflinks} translates to the following result involving only the aforementioned matrices.

\begin{corollary}
Let $f$ be a function on the unit circle which Fourier coefficients verify $d_{k}=d_{-k}$. Given two partitions $\lambda$ and $\mu$, we denote the Toeplitz minor generated by $f$ and indexed by $\lambda$ and $\mu$ by
\begin{equation*}
    D_{N}^{\lambda,\mu} (f) = \det{\left( d_{j-\lambda_{j}-k+\mu_{k}}\right)_{j,k=1}^{N}},
\end{equation*}
as in \cite{BumpDiaconis}. We have 
\begin{align*}
    &\det{\left( d_{j-k}-d_{j+k} \right)_{j,k=1}^{N}} = \frac{1}{2^{N}}\sum_{\rho_{1},\rho_{2}\in R(N)}(-1)^{(|\rho_{1}|+|\rho_{2}|)/2}D_{N}^{\rho_{1},\rho_{2}}(f), \\
    &\det{\left( d_{j-k}+d_{j+k-2} \right)_{j,k=1}^{N}} = \frac{1}{2^{N-2}}\sum_{\tau_{1},\tau_{2}\in T(N)}(-1)^{(|\tau_{1}|+|\tau_{2}|)/2}D_{N}^{\tau_{1}\tau_{2}}(f), \\
    &\det{\left(d_{j-k}-d_{j+k-1}\right)_{j,k=1}^{N}} = \frac{1}{2^{N}}\sum_{\sigma_{1},\sigma_{2}\in S(N)}(-1)^{(|\sigma_{1}|+|\sigma_{2}|+p(\sigma_{1})+p(\sigma_{2}))/2}D_{N}^{\sigma_{1}\sigma_{2}}(f). \\
\end{align*}
The minors appearing in the right hand sides above fit in the Toeplitz matrix generated by $f$ of order $2N+1$, $2N$ and $2N-1$, respectively, and the sums have $2^{2N-1}$ different terms, as in theorem \ref{th.hopflinks}.
\end{corollary}

For example, taking $N=2$ in the first identity above we obtain the expansion
\begin{align*}
    2\begin{vmatrix} d_0-d_2 & d_1-d_3 \\ d_1-d_3 & d_0-d_4 \end{vmatrix} = &\begin{vmatrix} d_0 & d_1 \\ d_1 & d_0 \end{vmatrix} - \begin{vmatrix} d_2 & d_1 \\ d_3 & d_0 \end{vmatrix} + \begin{vmatrix} d_3 & d_0 \\ d_4 & d_1 \end{vmatrix} - \begin{vmatrix} d_1 & d_2 \\ d_4 & d_1 \end{vmatrix} \\ +& \begin{vmatrix} d_1 & d_0 \\ d_4 & d_3 \end{vmatrix} - \begin{vmatrix} d_0 & d_1 \\ d_3 & d_2 \end{vmatrix} + \begin{vmatrix} d_0 & d_3 \\ d_3 & d_0 \end{vmatrix} - \begin{vmatrix} d_3 & d_2 \\ d_4 & d_3 \end{vmatrix},
\end{align*}
where all the determinants in the right hand side above are minors of the Toeplitz matrix $(d_{j-k})_{j,k=1}^{5}$. Analogous computations lead to expansions of minors of Toeplitz$\pm$Hankel matrices as sums of minors of Toeplitz matrices (equivalently, expansions of twisted integrals over $Sp(2N)$, $SO(2N)$ or $SO(2N+1)$ in terms of twisted integrals over $U(N)$). However, the resulting expressions are rather cumbersome and we do not pursue this road further.

\subsection{Gessel-type identities}

Another possibility for expressing integrals over the classical groups in terms of symmetric functions is available, in the form of Schur function series. A well known example of this is the classical identity of Gessel for Toeplitz determinants \cite{Gessel}. This, as well as generalizations for Toeplitz$\pm$Hankel determinants and minors of these matrices, is the content of the next theorem.

Let us denote by $\mathfrak{s}^{\nu}_{G(N)}(x)$ the Schur, symplectic Schur or even/odd orthogonal Schur symmetric function indexed by the partition $\nu$ for $G(N)=U(N),Sp(2N),SO(2N),SO(2N+1)$ respectively, for this theorem only. We also denote here and in the following by $s_{\nu/\mu}$ the skew Schur polynomial indexed by the skew shape $\nu/\mu$, see \cite{MacDonald} for instance.

\begin{theorem} \label{thgessel}
Let $x=(x_{1},x_{2},\dots)$ be a set of variables, and consider the function
\begin{equation*}
    H(x;e^{i\theta}) = \prod_{j=1}^{\infty}\frac{1}{(1-x_{j}e^{i\theta})}.
\end{equation*}
The following Schur functions series expansions hold
\begin{align}
    &\int_{G(N)}H(x;U)dU = \sum_{l(\nu)\leq N} s_{\nu}(x)\mathfrak{s}^{\nu}_{G(N)}(x), \label{gessel} \\
    &\int_{G(N)}\chi^{\mu}_{G(N)}(U)H(x;U)dU = \sum_{l(\nu)\leq N}s_{\nu/\mu}(x)\mathfrak{s}^{\nu}_{G(N)}(x),  \label{gessel1} \\
    &\int_{G(N)}\chi^{\lambda}_{G(N)}(U^{-1})\chi^{\mu}_{G(N)}(U)H(x;U)dU = \begin{dcases} \sum_{l(\nu)\leq N} s_{\nu/\lambda}(x)s_{\nu/\mu}(x),\qquad G(N)=U(N), \\ \sum_{l(\nu)\leq N}\sum_{\kappa} b^{\kappa}_{\lambda\mu}s_{\nu/\kappa}(x)\mathfrak{s}^{\nu}_{G(N)}(x),\qquad \textrm{rest of }G(N), \end{dcases} \label{gessel2}
\end{align}
where the coefficients $b^{\kappa}_{\lambda\mu}$ can be expressed in terms of Littlewood-Richardson coefficients $c_{\sigma\tau}^{\lambda}$ \cite{MacDonald} by the following formula 
\begin{equation*}
    b^{\kappa}_{\lambda\mu}=\sum_{\sigma,\rho,\tau}c_{\sigma\tau}^{\lambda}c_{\rho\tau}^{\mu}c_{\sigma\rho}^{\kappa}.
\end{equation*}
The same expansions hold if one replaces $H$ by the function 
\begin{equation}  \label{functionE}
    E(x;e^{i\theta}) = \prod_{j=1}^{\infty}(1+x_{j}e^{i\theta}),
\end{equation}
after transposing the partitions indexing all the symmetric functions in the above identities.
\end{theorem}

We remark the fact that the choice of functions above is without loss of generality. Indeed, recall that the Fourier coefficients of the functions $H(x;e^{i\theta})$ and $E(x;e^{i\theta})$ are the complete homogeneous symmetric functions $h_{k}(x)$ and the elementary symmetric functions $e_{k}(x)$ \eqref{ek} respectively. Both of these families are algebraically independent, and thus one can specialize them to any given values to recover any function with arbitrary Fourier coefficients from $H(x;e^{i\theta})$ or $E(x;e^{i\theta})$.

A similar proof of identity \eqref{gessel} for $G(N)=Sp(2N),SO(2N)$ can be found in \cite{Betea}. See also \cite{IW98}-\cite{Balantekin:2001id} for earlier related results. Different Schur function series for the integrals \eqref{gessel} can also be found in \cite{BaikRains,vandeLeur:2016yuq}.

\begin{proof}
The expansion \eqref{gessel} for $G(N)=U(N)$ is the aforementioned result of Gessel \cite{Gessel}, which extends easily to the other groups. We sketch the proof for convenience of the reader. Denote by $T(f)$ the infinite Toeplitz matrix generated by a function $f$. It is well known that if two functions $a,b$ satisfy
\begin{equation}  \label{a-b+}
    a(e^{i\theta}) = \sum_{k\leq 0} a_{k}e^{ik\theta},\qquad b(e^{i\theta}) = \sum_{k\geq 0} b_{k}e^{ik\theta} \qquad (z\in\mathbb{T})
\end{equation}
then the infinite Toeplitz matrix generated by the function $ab$ satisfies $T(ab) = T(a)T(b)$. It follows from the Cauchy-Binet formula that $\det{T_{N}(ab)}$ is then a sum over minors of the Toeplitz matrices of sizes $N\times\infty$ and $\infty\times N$ generated by $a$ and $b$, respectively, where $T_{N}(ab)$ denotes the Toeplitz matrix of size $N$ generated by $ab$. The proof is completed upon noting that if $a(e^{-i\theta})=b(e^{i\theta})=H(x;e^{i\theta})$ then by the Jacobi-Trudi identity \eqref{JTschur} the minors appearing in the sum are precisely the Schur polynomials appearing in \eqref{gessel}, since the Fourier coefficients of the function $H(x;e^{i\theta})$ are the complete homogeneous symmetric polynomials $h_{k}(x)$. The proof for the other groups is analogous: now the factorization
\begin{equation*}
    TH(ab) = T(a)TH(b)
\end{equation*}
holds for each of the Toeplitz$\pm$Hankel matrices $TH(b)$ appearing in \eqref{intSp2N}-\eqref{intO2N+1} and functions $a,b$ satisfying \eqref{a-b+}. The result then follows from the Cauchy-Binet formula and the Jacobi-Trudi identities \eqref{JTsymp}-\eqref{JToort} (some extra computation is needed in the odd orthogonal case, as in corollary \ref{cofacts}).

Identities \eqref{gessel1}, and \eqref{gessel2} for $U(N)$, follow analogously from the generalization of Jacobi-Trudi formula for skew Schur polynomials. Identity \eqref{gessel2} for the rest of the groups follows from \eqref{gessel1} and the fact that the characters $\chi_{G(N)}^{\lambda}$ follow the multiplication rule \cite{Littlewood58}
\begin{equation}
    \chi_{G(N)}^{\lambda}(U)\chi_{G(N)}^{\mu}(U)=\sum_{\nu}b_{\lambda\mu}^{\nu}\chi_{G(N)}^{\nu}(U)  \label{multrule}
\end{equation}
for $G(N)=Sp(2N),SO(2N)$ and $SO(2N+1)$ (recall that $\chi_{G(N)}^{\lambda}(U)=\chi_{G(N)}^{\lambda}(U^{-1})$ for such groups).

The corresponding identities involving the function $E$ follow analogously, using the dual Jacobi-Trudi identities instead (or, equivalently, using the involution $h_{k}\mapsto e_{k}$) in \eqref{gessel}-\eqref{gessel2}).
\end{proof}

We will be interested in the following in computing the $N\rightarrow\infty$ limit of the integrals $\int_{G(N)}f(U)dU$. This can be achieved by means of the strong Szeg\H o limit theorem and its generalization to the rest of the groups $G(N)$ due to Johansson \eqref{szego}-\eqref{szegooort}, or equivalently, by means of theorem \ref{thgessel} and the Cauchy identities \eqref{cauchyun}-\eqref{cauchyoort} (see section \ref{s.CS} below for such explicit computations). It turns out that the twisted integrals with characters on the integrand share a common asymptotic behavior.

\begin{theorem} \label{t.avg}
The averages of characters over any of the groups $G(N)$ satisfy 
\begin{equation}
    \lim_{N\rightarrow\infty}\frac{\int_{G(N)}\chi^{\lambda}_{G(N)}(U^{-1})\chi^{\mu}_{G(N)}(U)H(x;U)dU}{\int_{G(N)}H(x;U)dU} = \sum_{\nu}s_{\lambda/\nu}(x)s_{\mu/\nu}(x).  \label{averagesN}
\end{equation}
\end{theorem}

Note that if there is only one character in the integrand above the right hand side simplifies to a single Schur polynomial. As before, the theorem also holds for the function $E(x;e^{i\theta})=\prod_{j=1}^{\infty}(1+x_{j}e^{i\theta})$, after transposing the partitions indexing the skew Schur polynomials above.

\begin{proof}
If $G(N)=U(N)$, the result (that appeared first in \cite{GGT}) is a consequence of \eqref{gessel2} and the identity \cite{MacDonald} 
\begin{equation*}
    \sum_{\nu}s_{\nu/\lambda}(x)s_{\nu/\mu}(x) = \sum_{\nu}s_{\nu}(x)s_{\nu}(x)\sum_{\nu}s_{\lambda/\nu}(x)s_{\mu/\nu}(x),
\end{equation*}
where the sums run over all partitions $\nu$.

Suppose now that $G(N)=Sp(2N),SO(2N),SO(2N+1)$, and start by considering a single character in the integral. Then, using the Cauchy identity \eqref{cauchyun} and the restriction rules \eqref{restschursp}-\eqref{restschuroort} we obtain 
\begin{align*}
    \int_{G(N)}\chi^{\mu}_{G(N)}(U)H(x;U)dU = \sum_{l(\nu)\leq N}\sum_{\alpha}\sum^{\sim}_{\beta}c^{\nu}_{\alpha\beta}s_{\nu}(x)\int_{G(N)}\chi^{\mu}_{G(N)}(U)\chi^{\alpha}_{G(N)}(U)dU,
\end{align*}
where $\sum\limits^{\sim}$ denotes that the sum on $\beta$ runs over all even partitions for $G(N)=SO(2N),SO(2N+1)$, and over all partitions whose conjugate is even, for $G(N)=Sp(2N)$ (we say that a partition is even if it has only even parts), and the sum on $\alpha$ runs over all partitions. Taking $N\rightarrow\infty$ in the above expression and using the orthogonality of the characters with respect to Haar measure we obtain
\begin{equation}  \label{averagesNsingle}
    \lim_{N\rightarrow\infty}\int_{G(N)}\chi^{\mu}_{G(N)}(U)H(x;U)dU = s_{\mu}(x)\sum_{\beta}^{\sim}s_{\beta}(x).
\end{equation}
This gives the desired result upon noting that the sum on the right hand side is precisely the $N\rightarrow\infty$ limit of the integral $\int_{G(N)}H(x;U)dU$. The result for the integral \eqref{averagesN} twisted by two characters then follows from \eqref{averagesNsingle} and the multiplication rule \eqref{multrule}.
\end{proof}

In particular, we see that the $N\rightarrow\infty$ limit of the average is independent of the particular group $G(N)$ considered. This was noted in \cite{Dehaye} for a single character, and while this automatically implies the same for two characters for $G(N)=Sp(2N),SO(2N),SO(2N+1)$ (recall that $\chi_{G(N)}^{\lambda}(U^{-1})=\chi_{G(N)}^{\lambda}(U)$ for these groups), this is not immediate for $G(N)=U(N)$.

Note also that no mention of the regularity of the function $f$ has been made in the proof of theorem \ref{t.avg}. Indeed, only standard tools from the theory of symmetric functions are needed in order to obtain the result. This implies that the conclusion of the corollary holds for any integrable function, in particular for functions with Fisher-Hartwig singularities \cite{DIKrev}. We thus see that the possible change of behaviour in the large $N$ limit only affects the integrals $\int_{G(N)}f(U)dU$, and has no effect on the averaged integrals \eqref{averagesN}. See \cite{GGT} for more details on this.

\section{The case of Gaussian entries or $f(z)=\Theta (z)$}

We particularize the previous result to the case of a completely solvable model, for both finite and large $N$. It turns out to be related to many subjects: $G(N)$ Chern-Simons theory on $S^{3}$, the skein of the annulus and Hopf links. The corresponding Toeplitz and Toeplitz$\pm$Hankel matrices also appear in other contexts, as they are Fourier and sine/cosine transforms matrices.

\subsection{Partition functions of Chern-Simons theory on $S^{3}$}

\label{s.CS}

Let $q$ be a parameter satisfying $|q|<1$, and consider Jacobi's third theta function
\begin{align}
    &\sum_{n\in\mathbb{Z}} q^{n^{2}/2}e^{in\theta} = (q;q)_{\infty}\prod_{k=1}^{\infty}(1+q^{k-1/2}e^{i\theta})(1+q^{k-1/2}e^{-i\theta}), \label{fourtheta}
\intertext{where $(q;q)_{\infty}=\prod_{j=1}^{\infty}(1-q^{j})$. We then define $f(U)$ for $U\in G(N)$ as in \eqref{fprod}, with $f$ being the function}
    &\Theta(e^{i\theta}) = E(q^{1/2},q^{3/2},\dots;e^{i\theta}),  \label{thetaf}
\end{align}
where $E$ is given by \eqref{functionE}. For this choice of function, the integral 
\begin{align*}
    &Z_{G(N)} = (q;q)_{\infty}^{N}\int_{G(N)}\Theta(U)dU
\intertext{recovers the partition function of Chern-Simons theory on $S^{3}$ with symmetry group $G(N)$, and the coefficients in the corresponding Toeplitz and Toeplitz$\pm$Hankel matrices are $d_{k}=q^{k^{2}/2}$, according to \eqref{fourtheta}. Moreover, the averages}
    &\langle W_{\mu}\rangle_{G(N)} = \frac{1}{Z_{G(N)}}\int_{G(N)}\chi _{G(N)}^{\mu }(U)\Theta(U)dU
\intertext{and}
    &\langle W_{\lambda\mu}\rangle_{G(N)} = \frac{1}{Z_{G(N)}} \int_{G(N)}\chi _{G(N)}^{\lambda }(U^{-1})\chi _{G(N)}^{\mu }(U)\Theta(U)dU,
\end{align*}
where $l(\lambda),l(\mu)\leq N$, are, respectively, the Wilson loop and Hopf link of the theory. As we will see below, these matrix models are exactly solvable, and the formalism of Toeplitz and Toeplitz$\pm$Hankel determinants and minors allows an elementary and unified approach for their computation.

\subsubsection{Unitary group}

We start by reviewing the simplest and well-known case. We obtain from the
determinant expression \eqref{intUN} 
\begin{align*}
Z_{U(N)} =\det {(q^{(j-k)^{2}/2})_{j,k=1}^{N}} =q^{\sum_{j=1}^{N}j^{2}}\det {%
(q^{-jk})_{j,k=1}^{N}} =
\prod_{j<k}(1-q^{k-j})=\prod_{j=1}^{N-1}(1-q^{j})^{N-j},
\end{align*}
where the second identity follows from the fact that the second determinant
above is essentially the determinant of the matrix $M_{U(N)}(z)$
\eqref{detun}, with $z_{j}=q^{j-1}$.

The large-$N$ limit of this expression is given by Szeg\H o's theorem %
\eqref{szego}, which shows that as $N\rightarrow\infty$ 
\begin{align*}
Z_{U(N)} \sim \exp {\left(-N\sum_{k=1}^{\infty}\frac{1}{k}\frac{q^{k}}{%
1-q^{k}}+\sum_{k=1}^{\infty }\frac{1}{k}\frac{q^{k}}{(1-q^{k})^{2}}\right) }.
\label{Nunitszego}
\end{align*}
The same formula can be obtained using Cauchy's identity \eqref{cauchyun} in
formula \eqref{gessel}, as noted in \cite{TWwords}.

\subsubsection{Symplectic group}

We can proceed analogously for the rest of the groups. The determinants will
now be specializations of the corresponding matrix $M_{G(N)}(z)$ with $z_{j}
= q^{j}$, which can be computed explicitly by means of the formulas %
\eqref{detun}-\eqref{detone}. For the symplectic group we obtain 
\begin{align*}
Z_{Sp(2N)} &= \det\left(q^{(j-k)^{2}/2}-q^{(j+k)^{2}/2}\right)_{j,k=1}^{N}
=q^{\sum_{j=1}^{N}j^{2}}\det(q^{-jk}-q^{jk})_{j,k=1}^{N} \\
&=\prod_{j=1}^{N-j}(1-q^{j})^{N-j}\prod_{j=3}^{N}(1-q^{j})^{[\frac{j-1}{2}%
]}\prod_{j=N+1}^{2N-1}(1-q^{j})^{[\frac{2N+1-j}{2}]}%
\prod_{j=1}^{N}(1-q^{2j}) = \prod_{j=1}^{2N}(1-q^{j})^{\epsilon(j)},
\end{align*}
where 
\begin{equation*}
\epsilon(j) = \begin{dcases} N-\frac{j}{2}-\frac{1}{2},\qquad
&j\textrm{ odd}\,1\leq j\leq N, \\ N-\frac{j}{2},\qquad &j\textrm{ even},\,
1\leq j\leq N,\\ N-\frac{j}{2}+\frac{1}{2},\qquad &j\textrm{ odd},\,N+1\leq
j\leq 2N, \\ N-\frac{j}{2}+1,\qquad &j\textrm{ even},\, N+1\leq j\leq 2N.
\end{dcases}
\end{equation*}
As with the unitary model, this result is exact and holds for every $N$, and
coincides with the expression obtained in \cite{Sinha:2000ap} for the large $N$
regime. We see that the partition function of the symplectic model is
obtained as the product of the partition function of the unitary model and
extra factors.

For the large-$N$ limit, we obtain from Johansson's generalization of Szeg%
\H{o}'s theorem \eqref{szegosymp} that as $N\rightarrow \infty$ 
\begin{align*}
Z_{Sp(2N)} \sim \exp {\left( -N\sum_{k=1}^{\infty }\frac{1}{k}\frac{q^{k}}{%
1-q^{k}}+\frac{1}{2}\sum_{k=1}^{\infty }\frac{1}{k}\frac{q^{k}}{(1-q^{k})^{2}%
}+\sum_{k=1}^{\infty }\frac{1}{2k}\frac{q^{k}}{1-q^{2k}}\right)}.
\end{align*}
Again, the same result is obtained using Cauchy's identity for symplectic
characters \eqref{cauchysymp} in equation \eqref{gessel}. Notice that in the
large $N$ limit, the partition function for the $Sp(2N)$ model is a factor
of the partition function of the $U(N)$ model, while precisely the opposite
occurred at finite $N$.

\subsubsection{Orthogonal groups}

Proceeding analogously, we see that by identity \eqref{detone} 
\begin{align*}
Z_{SO(2N)} &= \frac{1}{2}\det{\left(q^{(j-k)^{2}/2}+q^{(j+k-2)^{2}/2}%
\right)_{j,k=1}^{N}} \\
&=\prod_{j=1}^{N-1}(1-q^{j})^{N-j}\prod_{j=1}^{N-1}(1-q^{j})^{[\frac{j+1}{2}%
]}\prod_{j=N}^{2N-3}(1-q^{j})^{[\frac{2N-j-1}{2}]} =
\prod_{j=1}^{2N-3}(1-q^{j})^{\epsilon(j)},
\end{align*}
where 
\begin{equation*}
\epsilon(j)=\begin{dcases} N-\frac{j}{2}+\frac{1}{2},\qquad &j\textrm{
odd},\,1\leq j \leq N-1, \\ N-\frac{j}{2},\qquad &j\textrm{ even},\, 1\leq
N-1, \\ N-\frac{j}{2}-\frac{1}{2},\qquad &j\textrm{ odd},\, N\leq j \leq
2N-3, \\ N-\frac{j}{2}-1,\qquad &j\textrm{ even},\, N\leq j \leq 2N-3,
\end{dcases}
\end{equation*}
in agreement with \cite{Sinha:2000ap}. Again, the partition function contains
as a factor the partition function of the unitary model. For $SO(2N+1)$ we
have 
\begin{align*}
Z_{SO(2N+1)}&=\det {\left(q^{(j-k)^{2}/2}-q^{(j+k-1)^{2}/2}\right)
_{j,k=1}^{N}} \\
&= \prod_{j=1}^{N-1}(1-q^{j})^{N-j}\prod_{j=2}^{N}(1-q^{j})^{[\frac{j}{2}%
]}\prod_{j=N+1}^{2N-2}(1-q^{j})^{[\frac{2N-j}{2}]}%
\prod_{j=1}^{N}(1-q^{j-1/2}) \\
&= \prod_{j=1}^{2N-2}(1-q^{j})^{\epsilon(j)}\prod_{j=1}^{N}(1-q^{j-1/2}),
\end{align*}
where 
\begin{equation*}
\epsilon(j)=\begin{dcases} N-\frac{j}{2}-\frac{1}{2},\qquad &j\textrm{ odd},
1\leq j \leq 2N-2, \\ N-\frac{j}{2},\qquad &j\textrm{ even}, 1\leq j \leq
2N-2, \end{dcases}
\end{equation*}
in agreement with \cite{Sinha:2000ap}. We see once again that the partition
function can be seen as the partition function of the unitary model times an
extra factor. In this case, also factors with half-integer exponents $%
(1-q^{j/2})$ are present.

Let us also record here the value of the closely related integral \eqref{intO2N+1-} for this choice of function, for completeness. We have
\begin{align*}
    &(q;q)_{\infty}^{N}\int_{SO(2N+1)}\Theta(-U)dU  = \prod_{j=1}^{2N-3}(1-q^{j})^{\epsilon(j)}\prod_{j=1}^{N}(1+q^{j-1/2}) = Z_{SO(2N+1)} \prod_{j=1}^{N}\frac{(1+q^{j-1/2})}{(1-q^{j-1/2})},
\end{align*}
where $\epsilon(j)$ is as in $Z_{O(2N+1)}$.

For the large-$N$ limit, we obtain from Johansson's theorem \eqref{szegoeort},\eqref{szegooort} that as $N\rightarrow \infty$,

\begin{align*}
Z_{SO(2N)}& \sim \exp {\left( -N\sum_{k=1}^{\infty }\frac{1}{k}\frac{q^{k}}{%
1-q^{k}}+\frac{1}{2}\sum_{k=1}^{\infty }\frac{1}{k}\frac{q^{k}}{(1-q^{k})^{2}%
}-\sum_{k=1}^{\infty }\frac{1}{2k}\frac{q^{k}}{1-q^{2k}}\right) }, \\
Z_{SO(2N+1)}& \sim \exp {\left( -N\sum_{k=1}^{\infty }\frac{1}{k}\frac{q^{k}}{%
1-q^{k}}+\frac{1}{2}\sum_{k=1}^{\infty }\frac{1}{k}\frac{q^{k}}{(1-q^{k})^{2}%
}-\sum_{k=1}^{\infty }\frac{1}{2k-1}\frac{q^{k-1/2}}{1-q^{2k-1}}\right) }.
\end{align*}

One can verify directly from the expressions obtained that in the large $N$ limit we recover the partition function of $U(N)$ as the product of the partition functions of $Sp(2N)$ and $SO(2N)$, consistently with corollary \ref{thdets}.

\subsection{Gross-Witten-Wadia model}

We have seen in theorem \ref{thdets} that there is a non-trivial factorization property of matrix integrals. This identity is independent of the choice of function and thus hence is applicable to other models, such as the Gross-Witten-Wadia model \cite{Gross:1980he,Wadia:2012fr}. This is interesting in view of new interest and results on the model \cite{Okuyama:2017pil,Ahmed:2017lhl,Ahmed:2018gbt,Itoyama:2018gnh}.

Recall that the Gross-Witten-Wadia model is characterized by a symbol function 
\begin{equation*}
    f_{GWW}(z)=\exp \left( -\beta\left( z+z^{-1}\right) \right).
\end{equation*}
In particular, the third identity in theorem \ref{thdets} allows to translate results on the much more widely studied case of the unitary group to the $SO(2N+1)$ case. More explicitely, if we denote $Z_{G(N)}^{GWW}(\beta) = \int_{G(N)}f_{GWW}(U)dU$, we have
\begin{equation*}
    Z_{U(2N)}^{GWW}\left( \beta \right) = Z_{SO(2N+1)}^{GWW}\left( \beta \right) Z_{SO(2N+1)}^{GWW}\left( -\beta \right) 
\end{equation*}
Other relationships can be obtained. For example, the first and last identities in theorem \ref{thdets}, together with \eqref{szegosymp} and \eqref{szegoeort}, show that at large $N$
\begin{equation*}
    Z_{U(2N-1)}^{GWW}(\beta),Z_{U(2N)}^{GWW}(\beta) \sim Z_{Sp(2N)}^{GWW}(\beta)Z_{SO(2N)}^{GWW}(\beta).
\end{equation*}
Likewise, it follows from the Szeg\H o-Johannson theorem quoted in Appendix B that at large $N$
\begin{equation*}
    Z_{U(2N)}^{GWW}\left( \beta \right) = (Z_{SO(2N)}^{GWW}\left( \beta \right))^2 = (Z_{Sp(2N)}^{GWW}\left( \beta \right))^2.
\end{equation*}%
This relationship also has a $XX$ spin chain interpretation \cite{Viti:2016hty}. This is however modified in the usual double scaling limit \cite{Myers:1990pp,Forrester:2012xu}. At any rate, it seems that large $N$ results for the GWW model \cite{Ahmed:2017lhl,Ahmed:2018gbt} can be translated to the $SO(2N)$ and $Sp(2N)$ models. It would also be interesting to further use this relationship between partition functions, by taking into account the well-known connection of $Z_{U(2N)}^{GW}$ with Painlevé equations \cite{ForrWitte},\cite{ForrWitte2},\cite{Itoyama:2018gnh}.



\section{Insertion of characters, minors, modular matrices and Hopf link
expansions}

We now turn to computing Wilson loops and Hopf links of Chern-Simons theory
on $S^{3}$ with symmetry group $G(N)$, for each of the classical groups. Let us fix two partitions $\lambda$ and $\mu$ of lengths $l(\lambda),l(\mu)\leq N$ throughout the rest of the section.

\subsection{Unitary group}

The insertion of a Schur polynomial on the unitary model gives 
\begin{align*}
(q;q)_{\infty}^{N}\int_{U(N)}s_{\mu }(U)\Theta (U)dU =
\det(q^{(j-k-\mu_{k}^{r})^{2}/2})_{j,k=1}^{N} =
q^{\sum_{j=1}^{N}\left(\mu_{j}^{2}/2+(N-j+1)\mu_{j}+j^{2}\right)}\det%
\left(q^{-j(k+\mu_{k}^{r})}\right)_{j,k=1}^{N}.
\end{align*}
We see that the determinant in the right hand side above is now essentially
the minor $M_{U(N)}^{\mu}(z)$ in \eqref{unchar} after setting $z_{j}=q^{-j}$%
. This leads to 
\begin{equation}
\langle W_{\mu}\rangle_{U(N)} =
q^{\sum_{j=1}^{N}\mu_{j}\left(\mu_{j}/2-j+1\right)}s_{\mu}(1,q,%
\dots,q^{N-1}),  \label{unitwilson}
\end{equation}
which, up to a prefactor of a power of $q$, recovers the original result in 
\cite{Dolivet:2006ii}. We recall that the above specialization of the Schur
polynomial is a polynomial on $q$ with positive and integer coefficients 
\cite{MacDonald}.

Inserting two Schur polynomials in the integral we obtain%
\begin{align*}
& (q;q)_{\infty }^{N}\int_{U(N)}s_{\lambda }(U^{-1})s_{\mu }(U)\Theta
(U)dU=\det (q^{(j+\lambda _{j}^{r}-k-\mu _{k}^{r})^{2}/2})_{j,k=1}^{N} \\
& =q^{\sum_{j=1}^{N}\left( \lambda _{j}^{2}/2+\mu _{j}^{2}/2+(N-j)(\lambda
_{j}+\mu _{j})+(j-1)^{2}\right) }\det (q^{-(N-j+\lambda _{j})(N-k+\mu
_{k})})_{j,k=1}^{N}.
\end{align*}%
The determinant is now a minor of $M_{U(N)}^{\lambda }(z)$, obtained by
striking some of its rows. That is, a minor obtained by striking rows 
\textit{and} columns of the Vandermonde matrix $M_{U(N)}(1,q,\dots,q^{N-1})$, as noted in \cite
{MortonLukac}. One can express this in terms of Schur polynomials by setting 
$z_{j}=q^{N-j+\mu _{j}}$ in this matrix, which yields 
\begin{align}
    \langle W_{\lambda\mu}\rangle_{U(N)}=q^{\sum_{j=1}^{N}\left( \lambda_{j}^{2}/2+\mu_{j}^{2}/2-(j-1)(\lambda _{j}+\mu _{j})\right) }&s_{\mu }(1,q,\dots ,q^{N-1})s_{\lambda }(q^{-\mu _{1}},q^{1-\mu _{2}},\dots ,q^{N-1-\mu _{N}}). \nonumber \\
\intertext{The above expression can also be written in terms of the quadratic Casimir element of $U(N)$, which we denote by $C_{2}^{U(N)}(\lambda) = \sum_{j}\lambda_{j}(\lambda_{j}+N-2j+1)$, as follows}
    q^{\left(-(N-1)(|\lambda |+|\mu |)+C_{2}^{U(N)}(\lambda )+C_{2}^{U(N)}(\mu)\right)/2}&s_{\mu }(1,q,\dots ,q^{N-1})s_{\lambda }(q^{-\mu _{1}},q^{1-\mu_{2}},\dots ,q^{N-1-\mu _{N}}). \label{unithopf}
\end{align}
Further interest in the minors of the Vandermonde matrix $M_{U(N)}(1,q,\dots,q^{N-1})$ and the rest of the matrices $M_{G(N)}$ arises from their relation with Chebotar\"ev's theorem\footnote{The matrix $M_{U(N)}(1,q,\dots,q^{N-1})$, for $q$ a $p$-th root of unity, is the matrix associated to the discrete Fourier transform (DFT), and Chebotar\"ev's classical theorem \cite{Cheb} states that every minor of this matrix is nonzero if $p$ is prime. An analogue of this theorem for the matrices of the discrete sine and cosine transforms, which correspond to $M_{Sp(2N)}(q,\dots,q^{N})$ and $M_{O(2N)}(1,\dots,q^{N-1})$ respectively, has been proved recently \cite{ChebGN}.} and the recent related advances in the topic \cite{ChebGN}.
 
We also see that a phenomenon already present when computing the partition
functions takes place when computing averages of Schur polynomials. For the
theta function, integrating the determinant $\det{M_{G(N)}(z)}$ in \eqref{intGN}
amounts essentially to computing the determinant of the matrix $M_{G(N)}(z)$
itself, after a certain specialization of the variables $z$. We also see
that the average of one or two Schur polynomials is expressed precisely as
the corresponding Schur polynomials, after some specialization to the same
number of nonzero variables as the size of the model.

This property has been noted in \cite{Mironov:2017och,Morozov:2018eiq} for
models of Hermitian Gaussian matrices. One point of view regarding these models, put forward in \cite{Morozov:2018eiq},
is that a main aspect of Gaussian matrix measures is that
they preserve Schur functions. We shall see that the same property
holds when changing the symmetry of the ensemble from unitary to symplectic
or orthogonal, by simply replacing Schur polynomials by symplectic or
orthogonal Schur functions.

\subsection{Symplectic group}

Performing analogous computations to the unitary case, we see that 
\begin{align}
&(q;q)_{\infty}^{N}\int_{Sp(2N)}\overline{sp_{\lambda }(U)}sp_{\mu
}(U)\Theta (U)dU =\det {(q^{(j+\lambda_{j}^{r}-k-\mu_{k}^{r})^{2}/2}-q^{(j+
\lambda_{j}^{r}+k+\mu_{k}^{r})^{2}/2})_{j,k=1}^{N}}  \nonumber \\
&=q^{\sum_{j=1}^{N}\left(
\lambda_{j}^{2}/2+\mu_{j}^{2}/2+(N-j+1)(\lambda_{j}+\mu _{j})+j^{2}\right)
}\det {(q^{-(j+\lambda _{j}^{r})(k+\mu_{k}^{r})}-q^{(j+\lambda
_{j}^{r})(k+\mu_{k}^{r})})_{j,k=1}^{N}},  \notag \\
\intertext{which leads to}
&\langle W_{\lambda\mu}\rangle_{Sp(2N)} = q^{\left(
N(|\lambda|+|\mu|)+C_{2}^{Sp(2N)}(\lambda)+C_{2}^{Sp(2N)}(\mu
)\right)/2}sp_{\mu}(q,q^{2},\dots ,q^{N})sp_{\lambda }(q^{1+\mu _{N}},\dots
,q^{N+\mu _{1}}),  \label{kauffsymp}
\end{align}
where we have identified $C_{2}^{Sp(2N)}(\lambda)=\sum_{j}\lambda_{j}(%
\lambda_{j}+N-2j+2)$, the quadratic Casimir element of $Sp(2N)$. As before,
the second identity in \eqref{kauffsymp} follows from the fact that
integrating the function $\Theta $ we recover a (row and column-wise) minor
of the matrix $M_{Sp(2N)}(z)$ itself, specialized to $z_{j}=q^{j}$. We note
that $\lambda $ and $\mu $ are interchangeable in the above formula, and
also that setting one of the partitions to be empty we obtain a formula for
the average of a single character $\langle W_{\mu}\rangle_{Sp(2N)}$.

\subsection{Orthogonal groups}

For the orthogonal models we have 
\begin{align}
&(q;q)_{\infty}^{N}\int_{SO(2N)}o_{\lambda}^{even}(U)o_{\mu }^{even}(U)\Theta
(U)dU=\frac{1}{2}\det {\left( q^{(j+\lambda
_{j}^{r}-k-\mu_{k}^{r})^{2}/2}+q^{(j+\lambda_{j}^{r}+k+\mu
_{k}^{r}-2)^{2}/2}\right) _{j,k=1}^{N}}  \notag \\
&=\frac{1}{2}q^{\sum_{j=1}^{N}\left( \lambda
_{j}^{2}/2+\mu_{j}^{2}/2+(N-j)(\lambda _{j}+\mu _{j})+(j-1)^{2}\right) }\det 
{\left(q^{-(N-j+\lambda _{j})(N-k+\mu _{k})}+q^{(N-j+\lambda
_{j})(N-k+\mu_{k})}\right) _{j,k=1}^{N}},  \notag \\
\intertext{which can be rewritten as}
&\langle W_{\lambda\mu}\rangle_{SO(2N)} = q^{\left(N(|\lambda
|+|\mu|)+C_{2}^{SO(2N)}(\lambda)+C_{2}^{SO(2N)}(\mu)\right)/2}o_{%
\mu}^{even}(1,q,\dots ,q^{N-1})o_{\lambda }^{even}(q^{\mu _{N}},\dots
,q^{N-1+\mu _{1}}),  \label{kauffeort}
\end{align}
where $C_{2}^{SO(2N)}(\lambda)=\sum_{j=1}^{N}\lambda_{j}(\lambda_{j}+N-2j)$
is the quadratic Casimir of $SO(2N)$. As before, setting one partition to be
empty we obtain a formula for the Wilson loop $\langle W_{\mu}\rangle_{SO(2N)}$. For the odd
orthogonal group $SO(2N+1)$ we obtain 
\begin{align}
    (q;q)_{\infty}^{N}\int_{SO(2N+1)}&o_{\lambda}^{odd}(U)o_{\mu}^{odd}(U)\Theta (U)dU=\det{\left(q^{(j+\lambda_{j}^{r}-k-\mu_{k}^{r})^{2}/2}-q^{(j+\lambda_{j}^{r}+k+\mu_{k}^{r}-1)^{2}/2}\right) _{j,k=1}^{N}}  \nonumber \\
    &=q^{\sum_{j=1}^{N}\left(\lambda_{j}^{2}/2+\mu_{j}^{2}/2+(N-j+1/2)(\lambda_{j}+\mu_{j})+(j-1/2)^{2}\right)} \nonumber \\
    &\times\det{\left(q^{-(N-j+\lambda_{j}+1/2)(N-k+\mu_{k}+1/2)}-q^{(N-j+\lambda_{j}+1/2)(N-k+\mu _{k}+1/2)}\right) _{j,k=1}^{N}},  \nonumber
\end{align}
which yields
\begin{align} \begin{split}
    \langle W_{\lambda\mu}\rangle_{SO(2N+1)} = &q^{\left((N+1/2)(|\lambda|+|\mu|)+C_{2}^{SO(2N+1)}(\lambda)+C_{2}^{SO(2N+1)}(\mu)\right)/2} \\
    &\times o_{\mu}^{odd}(q^{1/2},q^{3/2},\dots,q^{N-1/2})o_{\lambda}^{odd}(q^{1/2+\mu _{N}},q^{3/2+\mu_{N-1}},\dots,q^{N-1/2+\mu_{1}}), \label{kauffoort}
\end{split} \end{align}
with $C_{2}^{SO(2N+1)}(\lambda)=\sum_{j=1}^{N}\lambda_{j}(\lambda_{j}+N-2j+1/2)$ the quadratic Casimir of $SO(2N+1)$.

\subsection{Giambelli compatible processes}

The classical Giambelli identity expresses a Schur polynomial indexed by a general partition $\lambda$ as the determinant of a matrix which entries are Schur polynomials indexed only by hook partitions. More precisely
\begin{equation*}
    s_{(a_{1},\dots,a_{p}|b_{1},\dots,b_{p})}(x) = \det{(s_{(a_{j}|b_{k})}(x))_{j,k=1}^{p}},
\end{equation*}
where we have used the Frobenius notations for the partitions in the above identity (see the beginning of section \ref{s.minexp}). In \cite{GiamBor}, the notion of ``Giambelli compatible" processes was introduced to refer to probability measures on point configurations that preserve the Giambelli identity above, in the sense that
\begin{equation*}
    \langle s_{(a_{1},\dots,a_{p}|b_{1},\dots,b_{p})} \rangle = \det{(\langle s_{(a_{j}|b_{k})}\rangle )_{j,k=1}^{p}},
\end{equation*}
where the bracket notation $\langle s_{\lambda} \rangle$ denotes the average of the Schur polynomial $\lambda$ with respect the corresponding probability measure. Since then, several matrix models and gauge theories have been proved to be Giambelly compatible, including biorthogonal ensembles \cite{GiamTierz}, ABJM theory \cite{Hatsuda:2013yua}, and supersymmetric Chern-Simons theory \cite{Eynard:2014rba},\cite{Matsuno:2016jjp}.

Using the formulas obtained in the previous sections, one can easily prove that the random matrix ensembles corresponding to the theta function \eqref{thetaf} with $G(N)$ symmetry are Giambelli compatible in a slightly generalized sense. Indeed, we have seen that the average of a character over these ensembles can be evaluated as the precise same character, with a certain specialization, times a prefactor in the parameter $q$ (equations \eqref{unitwilson},\eqref{kauffsymp},\eqref{kauffeort},\eqref{kauffoort}). This fact, together with the Giambelli identity for the characters of the groups $G(N)$ \cite{GiamGN,GiamKratt}
\begin{equation*}
    \chi_{G(N)}^{(a_{1},\dots,a_{p}|b_{1},\dots,b_{p})}(U) = \det{\left(\chi_{G(N)}^{(a_{j}|b_{k})}(U)\right)_{j,k=1}^{p}},
\end{equation*}
and some computations to take care of the prefactors, show that
\begin{equation*}
    \langle W_{(a_{1},\dots,a_{p}|b_{1},\dots.b_{p})} \rangle_{G(N)} = \det{\left( \langle W_{(a_{j}|b_{k})}\rangle_{G(N)} \right)_{j,k=1}^{N}}.
\end{equation*}
That is, the Giambelly identity is preserved, after replacing the Schur polynomials in both sides of the identity with the corresponding character $\chi^{\lambda}_{G(N)}$. For $G(N)=U(N)$ this is a known result, as we are considering an orthogonal polynomial ensemble (which were proven to be Giambelli compatible in \cite{GiamBor}). However, for the rest of the groups $G(N)$ this provides an example of an ensemble with non unitary symmetry that is Giambelli compatible.

\subsection{Large $N$ limit and Hopf link expansions}

The expansions found in theorem \ref{th.hopflinks} have particular consequences when considering the Chern-Simons model. Considering the function $\Theta$ in this theorem and taking into account the results in section \ref{s.CS}, we see that at finite $N$ the partition functions of $Sp(2N),SO(2N)$ and $SO(2N+1)$ Chern-Simons theories can be expressed as sums of unnormalized Hopf links of the unitary theory. On the other hand, theorem \ref{t.avg} implies that
\begin{align}
    &\lim_{N\rightarrow\infty}{\langle W_{\lambda\mu} \rangle_{G(N)}} =  \sum_{\nu}s_{(\lambda/\nu)^{\prime}}(q^{1/2},q^{3/2},\dots)s_{(\mu/\nu)^{\prime}}(q^{1/2},q^{3/2},\dots ) \label{qdimG2}
\end{align}
for each of the groups\footnote{The partitions in \eqref{averagesN} appear now conjugated, since the function is $\Theta$ is expressed as a specialization of $E(x;e^{i\theta})$.} $G(N)$, where the prime notation $'$ stands for the conjugated partition. Note that if there is only one character in the average the above formula simplifies to
\begin{align}
    &\lim_{N\rightarrow\infty} \langle W_{\mu} \rangle_{G(N)} = s_{\mu'}(q^{1/2},q^{3/2},\dots).  \label{qdimG}
\end{align}
Putting these two facts together we arrive at the following expansions
\begin{align*}
    \frac{Z_{Sp(2N)}}{Z_{U(N)}} &\sim \frac{1}{2^{N}}\sum_{\rho_{1},\rho_{2}\in   R(\infty)}(-1)^{(|\rho_{1}|+|\rho_{2}|)/2}\langle W_{\rho_{1}\rho_{2}}\rangle_{G(N)}, \\
    \frac{Z_{SO(2N)}}{Z_{U(N)}} &\sim \frac{1}{2^{N-1}}\sum_{\tau_{1},\tau_{2}\in T(\infty)}(-1)^{(|\tau_{1}|+|\tau_{2}|)/2}\langle W_{\tau_{1}\tau_{2}}\rangle_{G(N)} \\
    \frac{Z_{SO(2N+1)}}{Z_{U(N)}} &\sim \frac{1}{2^{N}}\sum_{\sigma_{1},\sigma_{2}\in S(\infty)}(-1)^{(|\sigma_{1}|+|\sigma_{2}|+p(\sigma_{1})+p(\sigma_{2}))/2}\langle W_{\sigma_{1}\sigma_{2}}\rangle_{G(N)}
\end{align*}
as $N\rightarrow\infty$, where the sets $R(\infty),S(\infty)$ and $T(\infty)$ are defined as the sets $R(N),S(N)$ and $T(N)$ respectively (see theorem \ref{th.hopflinks}) without the restriction $\alpha_{1}\leq N-1$. That is, at large $N$ the partition functions of the symplectic or orthogonal theories can be expressed as that of the unitary theory with an infinite number of corrections, which correspond to Wilson loops and Hopf links, indexed by partitions of increasing complexity\footnote{Note that the empty partition belongs to each of the sets $R(\infty),S(\infty)$ and $T(\infty)$, and thus the first term in the sums is always a $1$.} (and which are the same in this limit for each of the groups $G(N)$). Previous examples of partition functions of Chern-Simons theory expressed as sums of averages of characters can be found in \cite{Iqbal:2003ix}-\cite{Marino:2009mw}.

\section{Fermion quantum models with matrix degrees of freedom}

Some interest has arised recently in the study of fermionic quantum mechanical models with matrix degrees of freedom \cite{Anninos:2016klf,Tierz:2017nvl,Klebanov:2018nfp}. These models appear as specific instances of tensor quantum mechanical models \cite{Klebanov:2018nfp} and have a distinctive spectra of harmonic-oscillator type, but with exponentially degenerated energy levels, which suggests connections with other solvable models and to integrability. 

The fermionic model in \cite{Anninos:2016klf} consists in $NL$ complex fermions $%
\{\psi ^{iA},\bar{\psi}^{Ai}\}$ with $i=1,\ldots ,N$ and $A=1,\ldots ,L$. The indices $i$ and $A$ transform in the bifundamental of a $U(N)\times U(L)$ symmetry and the finite-dimensional Hilbert space (of dimension $2^{N\times L}$) consists in $N\times L$ fermionic operators satisfying $\{\bar{\psi}^{Ai},\psi ^{Bj}\}=\delta ^{ij}\delta
^{AB}$. The interactions are controlled by a Hamiltonian with quartic interactions and the normal ordering of the lower, quadratic interaction terms was not worked out in \cite{Anninos:2016klf}. The same model appeared in \cite{Klebanov:2018nfp} and the terms quartic and quadratic written down, but with different coefficients for the quadratic terms and hence with just slightly different numerical spectra, but with the same features. For an overall discussion of this model and the role of models with finite-dimensional Hilbert space at the level of AdS/CFT (or dS/CFT) type of descriptions, see \cite{Anninos:2020ccj}.

The spectrum of the model in \cite{Anninos:2016klf} is computed in \cite{Tierz:2017nvl}, based on the matrix model description obtained in \cite{Anninos:2016klf}. Likewise, in \cite{Klebanov:2018nfp}, spectra are also computed, using their identification of the Hamiltonian with quartic interactions in terms of Casimirs. In analogy with the form of the matrix model in \cite{Anninos:2016klf}, we compute here averages of insertions of characteristic polynomial type in the $G(N)$ Chern-Simons matrix model.

We emphasize this is simply in analogy with the model in \cite{Anninos:2016klf}, which described $U(N)\times U(L)$ fermion models in terms of the average of the $L$-th moment of a determinant insertion in $U(N)$
Chern-Simons matrix models. One motivation is that more complex models than the one in \cite{Anninos:2016klf,Tierz:2017nvl}, with symmetries such as $SO(N)\times SO(L)$, are
given in \cite{Klebanov:2018nfp} with qualitatively the same spectra, after numerically diagonalizing, in this case, the Hamiltonian. The models we study correspond to study the average of the function
\begin{align}
    &\Theta^{(L,m)}(e^{i\theta}) = \left(2\cos{\frac{\theta+im}{2}}\right)^{L}\Theta(e^{i\theta}) \nonumber
\intertext{over the groups $G(N)$, where $L$ is a positive integer and $m$ is a real parameter. In sight of \eqref{fprod} and the identity $2\cos\frac{\theta}{2}=|1+e^{i\theta}|$, we see that for $U$ belonging to any of the groups $G(N)$ we have}
    &\Theta^{(L,m)}(U) = \Theta(U)e^{Lm}\prod_{j=1}^{N} (1+e^{-m}e^{i\theta_{j}})^{L}(1+e^{-m}e^{-i\theta_{j}})^{L}, \label{thetaLm}
\intertext{where the $e^{i\theta_{j}}$ are the nontrivial eigenvalues of $U$. We will denote this average by}
    &Z^{(L,m)}_{G(N)} = \frac{1}{Z_{G(N)}}\int_{G(N)}\Theta^{(L,m)}(U)dU. \nonumber
\end{align}
Taking the limit $m\rightarrow0$ of the unitary model $Z^{(L,m)}_{U(N)}$ we recover the compactly supported analogue of the model considered in \cite{Tierz:2017nvl}. This model is also related with the Ewens measure on the symmetric group, see \cite{Olshanski} for instance.

\subsection{Unitary group}

Using the dual Cauchy identity \eqref{dcauchyun} twice to expand the product in \eqref{thetaLm} and identity \eqref{unithopf} we obtain 
\begin{align} \begin{split} \label{unitLm}
    Z^{(L,m)}_{U(N)} = e^{Lm}&\sum_{\lambda,\mu} s_{\lambda'}(\underbrace{e^{-m},\dots,e^{-m}}_{L})s_{\mu'}(\underbrace{e^{-m},\dots,e^{-m}}_{L})\langle W_{\lambda\mu} \rangle_{U(N)}  \\
    = e^{Lm}&\sum_{\lambda,\mu} e^{-m(|\lambda|+|\mu|)}s_{\lambda'}(1^{L})s_{\mu'}(1^{L})q^{(C_{2}^{U(N)}(\lambda)+C_{2}^{U(N)}(\mu))/2} \\
    &\times s_{\mu }(1,q^{-1},\dots ,q^{-(N-1)})s_{\lambda}(q^{-\mu_{N}},q^{-(\mu_{N-1}+1)},\dots ,q^{-(\mu_{1}+N-1)}),
\end{split} \end{align}
where $1^{L}$ denotes the specialization $x_{1}=\dots =x_{L}=1$. Recall that an explicit formula for $s_{\mu}(1^{L})$ is available \cite{MacDonald}. Now, since $s_{\nu}(x_{1},\dots,x_{N})=0$ if $l(\nu)>N$, we see that the above sum is actually over all partitions $\lambda,\mu$ contained in the rectangular diagram\footnote{See \cite{Panova} for recent results on asymptotics on the number of such partitions as $L$ and $N$ grow to infinity.} $(L^{N})$. Several nontrivial features of the model can be deduced from this fact.

First of all, we see that $Z^{(L,m)}_{U(N)}$ is a polynomial on $q^{1/2}$ and $e^{-m}$. The high number of terms in this polynomial compared to its relatively low degree on $q$ implies the high number of degeneracies in the spectrum mentioned above. Figure \ref{f.UL1} shows some examples where this phenomenon is apparent. Secondly, using the dual Cauchy identity again we see that in the limit $q\rightarrow1$ we have
\begin{equation*}
    \lim_{q\rightarrow1}Z^{(L,m)}_{U(N)} = e^{Lm}(1+e^{-m})^{2NL}.
\end{equation*}
Up to the prefactor $e^{Lm}$, this shows the duality between the parameters $(N,L)$ in this limit \cite{Tierz:2017nvl}. Finally, the expression \eqref{unitLm} allows direct computation of the model for low values of $N$ and $L$ and implementation in a computer algebra system. For instance, for $L=1$ we have 
\begin{equation*}
    \langle \Theta^{(L=1,m)} \rangle_{U(N)} = e^{m}\sum_{r,s=0}^{N}e^{-m(r+s)}q^{s-s^{2}/2+r/2}\qbinom{N}{r}e_{s}(q^{-1},1,q,\dots,q^{r-2},q^{r},q^{r+1},\dots,q^{N-1}),
\end{equation*}
where $e_{k}$ denotes the $k$-th elementary symmetric polynomial \eqref{ek}. 

\begin{center}
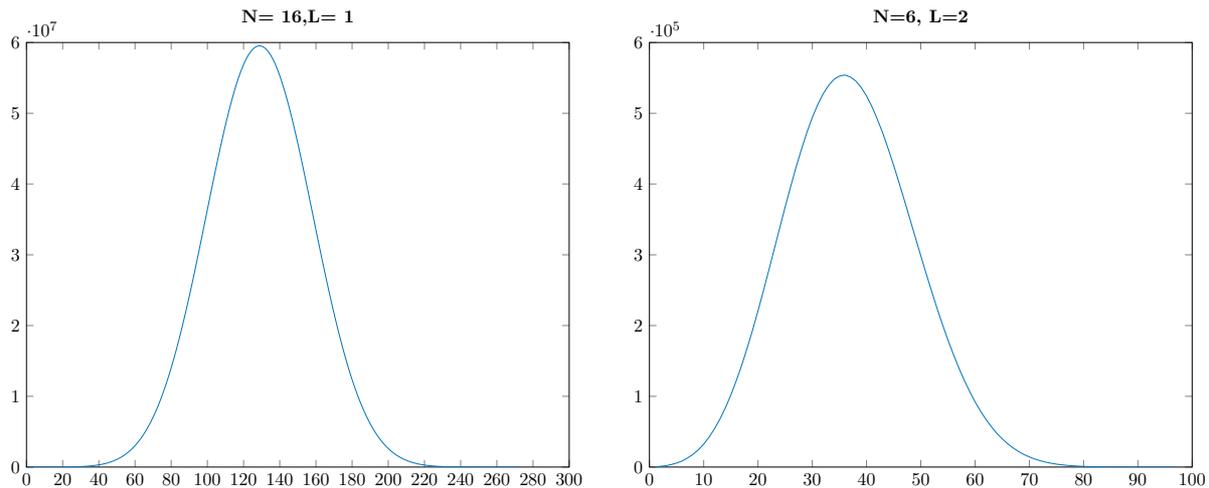
\begin{figure}
    \begin{scaletikzpicturetowidth}{.48\textwidth}
    \input{U2N16L1tex.tex}
    \end{scaletikzpicturetowidth}
    \begin{scaletikzpicturetowidth}{.48\textwidth}
    \input{U2N6L2tex.tex}
    \end{scaletikzpicturetowidth}
\caption{For each $n$ in the $x$ axis, the $y$ axis shows the coefficient of the monomial $q^{n/2}$ in $Z^{(L=1,m=0)}_{U(16)}$ (left) and $Z^{(L=2,m=0)}_{U(6)}$ (right).}
\label{f.UL1}
\end{figure}
\end{center}

\subsubsection*{Large-$N$ limit}
The large $N$ limit of the model can be computed by two different means, depending on the value of $m$. If $m$ is nonzero, it follows from 
\eqref{unitLm} and the identity \eqref{qdimG2} that
\begin{align*}
    \lim_{N\rightarrow\infty} Z^{(L,m)}_{U(N)} = e^{Lm}\sum_{\lambda,\mu} &s_{\lambda'}(\underbrace{e^{-m},\dots,e^{-m}}_{L})s_{\mu'}(\underbrace{e^{-m},\dots,e^{-m}}_{L}) \\
    &\times\sum_{\nu}s_{(\lambda/\nu)'}(q^{1/2},q^{3/2},\dots)s_{(\mu/\nu)'}(q^{1/2},q^{3/2},\dots) \\
    = e^{Lm}(1-&e^{-2m})^{-L^{2}}\prod_{k=1}^{\infty}\frac{1}{(1-e^{-m}q^{k-1/2})^{2L}},
\end{align*}
where the second identity above follows from standard manipulations of Schur and skew Schur polynomials\footnote{More precisely, we have used the expansion $s_{\lambda/\nu}=\sum_{\alpha}c^{\lambda}_{\nu\alpha}s_{\alpha}$, the multiplication rule $\sum_{\lambda}c^{\lambda}_{\nu\alpha}s_{\lambda} = s_{\nu}s_{\alpha}$ and the Cauchy identity \eqref{cauchyun}, where the $c^{\lambda}_{\nu\alpha}$ are Littlewood-Richardson coefficients.}.

The above expression is no longer valid in the massless case, $m=0$. Nevertheless, the large $N$ limit of the model can still be computed, using the fact that $Z^{(L,m)}_{U(N)}$ can be seen as the determinant of the Toeplitz matrix generated by the function $\Theta^{(L,m)}$ (recall identity \eqref{intUN}). For $m=0$, this function does not verify the hypotheses in Szeg\H{o}'s theorem, but it can be written as the product of a function that does verify these hypotheses (the function $\Theta$, as in section \ref{s.CS}) and a Fisher-Hartwig singularity. The asymptotic behaviour of Toeplitz determinants generated by such functions has been long studied \cite{DIKrev} and is now well understood \cite{DIK}. See appendix B for the definition of Fisher-Hartwig singularity and the relevant results that we will use in the following.

According to \eqref{pureFH}, we see that the function $\Theta ^{(L,m=0)}$ corresponds to the product of the smooth function $\Theta$ (in the sense of Szeg\H{o}'s theorem) and a single singularity at the point $z=-1$, with parameters $\alpha=L$ and $\beta=0$. This implies that as $N\rightarrow\infty$ we have \eqref{FHun} 
\begin{equation} \label{largeNFHun}
    Z^{(L,m=0)}_{G(N)} \sim N^{L^{2}}\frac{G^{2}(L+1)}{G(2L+1)}\prod_{k=1}^{\infty }
\frac{1}{(1-q^{k-1/2})^{2L}},
\end{equation}
where $G(z+1)$ is Barnes' $G$ function. Using its well known asymptotic expansion\footnote{For any $z$ in a sector not containing the negative real axis it holds that
\begin{equation*}
    \log G(z+1)=\frac{1}{12}-\log A+\frac{z}{2}\log 2\pi +\left( \frac{z^{2}}{2}-\frac{1}{12}\right) \log z-\frac{3z^{2}}{4}+\sum_{k=1}^{N}\frac{\mathrm{B}_{2k+2}}{4k(k+1)z^{2k}}+O\left( \frac{1}{z^{2N+2}}\right),
\end{equation*}
where $A$ is the Glaisher--Kinkelin constant and the $\mathrm{B}_{k}$ are the Bernouilli numbers.} we see that as $L\rightarrow\infty$ the free energy of the model satisfies
\begin{equation*}
    \lim_{L\rightarrow \infty } \log{ Z^{(L,m)}_{U(N\rightarrow\infty)} } \sim L^{2}\log \left( \frac{N}{L}\right) -L^{2}\left( 2\log 2-3/2\right) -\frac{\log L}{12}-2L\log\left( \sqrt{q},q\right) _{\infty },
\end{equation*}
where we have written the last term as a $q-$Pochhammer symbol\footnote{This type of piece also appears in the free energy of some $4d$ supersymmetric gauge theories \cite{Russo:2013qaa}.}. We have taken the large $L$ limit after the large $N$ limit. This is non-rigorous but standard in estimating free energies in the regime where one defines a Veneziano parameter\footnote{In analogy with localization, $L$ could be interpreted as number of flavours, but with hypermultiplets describing fermionic matter, and hence in the numerator in the matrix model. For example, in \cite{Betzios:2017yms} we see this type of insertions in the context of matrix quantum mechanics.} $\zeta=L/N$  and the double scaling is $\zeta=cte$ for $N\rightarrow \infty$ and $L\rightarrow \infty$. As we see, the leading term of the free energy vanishes for $\zeta=1$, and changes sign with $\zeta\rightarrow 1/\zeta$ otherwise.

Table \ref{t.FH} shows some numerical tests of the accuracy of formula \eqref{largeNFHun} (as well as the analogous formulas for the rest of the models, see the following subsections) for several values of $q$ and $N$.

\begin{table}
\begin{center}
\begin{tabular}{ |c|c|c|c|c|c| } 
    \hline
    Model & $N=4$ & $N=6$ & $N=8$ & Value of $q$ \\ \hline
    $U(N)$ & 1.0018 & 1.0005 & 1.0003 & $q=0.1$ \\
    $Sp(2N)$ & 0.9559 & 0.9692 & 0.9768 & $q=0.25$ \\
    $O(2N)$ & 0.9726 & 0.9970 & 0.9997 & $q=0.33$ \\
    $O(2N+1)$ & 0.8616 & 0.9631 & 0.9906 & $q=0.5$ \\
    \hline
\end{tabular}
\vspace*{7pt}
\label{t.FH}
\caption{The table shows the quotient between the numerical value of the spectrums $Z^{(L=1,m=0)}_{G(N)}$, computed directly by means of the formulas \eqref{unitLm},\eqref{sympLm},\eqref{eoLm},\eqref{ooLm}, and the predicted value given by formulas \eqref{largeNFHun},\eqref{FHsymp},\eqref{FHort}. The high rate of convergence is apparent already at low values of $N$. The rightmost column shows the value of $q$ at which the spectrum is computed.}
\end{center}
\end{table}

Let us emphasize that both the symmetric function approach and the Toeplitz determinant realization of the matrix model are useful for computing its large $N$ limit. Indeed, in the massive case, the character expansion is immediate and gives a manageable expression of the model, while the massless case is also readily handled with the aid of a particular example of Fisher-Hartwig asymptotics.

\subsection{Symplectic group}

We can proceed analogously for the rest of the groups $G(N)$. The expression resulting from the character expansion is actually simpler in this case, although some extra care needs to be taken before integrating. Let us start with the symplectic group. First, we use the dual Cauchy identity \eqref{dcauchysymp} to expand the product in \eqref{thetaLm}, obtaining
\begin{equation}
    Z^{(L,m)}_{Sp(2N)} = e^{Lm}(1-e^{-2m})^{-L(L+1)/2}\sum_{\mu}e^{-|\mu|m}s_{\mu'}(1^{L})\int_{Sp(2N)}sp_{\mu}(U)\Theta(U)dU. \label{sympLmaux}
\end{equation}
Since $sp_{\mu}(x_{1},\dots,x_{N})=0$ if $l(\mu)-\mu_{1}-1>2N$ (as can be seen from \eqref{JTsymp}, for instance), we see that the sum above actually runs over all partitions contained in the rectangular diagram $(L^{2N+L+1})$, and therefore is finite. However, we can only use formula \eqref{kauffsymp} and substitute the integral in \eqref{sympLmaux} by the Wilson loop $\langle W_{\mu}\rangle_{Sp(2N)}$ for those partitions satisfying $l(\mu)\leq N$. One can bypass this constraint in the following way. It is proven in \cite{KoikeTerada} (see proposition 2.4.1) that any $sp_{\mu}(U)$ (seen as a symmetric function, specialized to the nontrivial eigenvalues of $U$) indexed by a partition of length $l(\mu)> N$ either vanishes or coincides with an irreducible character $\chi^{\lambda}_{Sp(2N)}(U)$, with $l(\lambda)\leq N$, up to a sign. One can then substitute those $sp_{\mu}(U)$ in \eqref{sympLmaux} by the corresponding $\chi^{\lambda}_{Sp(2N)}(U)$, use formula \eqref{kauffsymp} to write the integrals as the Wilson loops $\langle W_{\lambda}\rangle_{Sp(2N)}$, and then undo the change to recover the $\langle W_{\mu}\rangle_{Sp(2N)}$ indexed by the original partition $\mu$ (recall that these coincide themselves with a symplectic Schur function, up to a prefactor). This yields the formula
\begin{equation}
    Z^{(L,m)}_{Sp(2N)} = e^{Lm}(1-e^{-2m})^{-L(L+1)/2}\sum_{\mu}e^{-|\mu|m}s_{\mu'}(1^{L})\langle W_{\mu} \rangle_{Sp(2N)}, \label{sympLm}
\end{equation}
where the sum runs over all partitions contained in the rectangular shape $(L^{2N+L+1})$. An analogous analysis to the unitary case is can be performed now. In particular, in the $q\rightarrow1$ limit we obtain
\begin{equation*}
    \lim_{q\rightarrow1}Z^{(L,m)}_{Sp(2N)} = e^{Lm}(1+e^{-m})^{2NL}
\end{equation*}
using the dual Cauchy identity \eqref{dcauchysymp}. Thus, not only does the $(N,L)$ duality hold for the symplectic group, up to the prefactor $e^{Lm}$, but the model is actually the same as the unitary one in the $q\rightarrow1$ limit.

Also as in the unitary case, the above sum gives rise to a highly degenerated spectrum. See figure \ref{f.SpO} for an example; explicit instances for lower values of $N$ and $L$ can also be computed easily. For instance, using the fact that $sp_{(1^{k})}(x_{1},\dots,x_{N})=-sp_{(1^{2N+2-k})}(x_{1},\dots ,x_{N})$ (which follows from \eqref{JTsymp}), we obtain for $L=1$ the expression
\begin{align*}
    Z_{Sp(2N)}^{(L=1,m)}&=e^{m}(1-e^{-2m})^{-1}\sum_{k=0}^{2N+2}e^{-km}q^{Nk+k-k^{2}/2}sp_{(1^{k})}(q,\dots,q^{N}) \\
    &=e^{m}(1-e^{-2m})^{-1}\sum_{k=0}^{N}e^{-km}(1-e^{-(N-k+1)2m})q^{Nk+k-k^{2}/2}sp_{(1^{k})}(q,\dots,q^{N}) \\
    & =e^{m}\sum_{k=0}^{N}e^{-km}(1+e^{-2m}+e^{-4m}+\dots+e^{-(N+k)2m})q^{Nk+k-k^{2}/2}sp_{(1^{k})}(q,\dots,q^{N}).
\end{align*}
We see that the prefactor $(1-e^{-2m})^{-1}$ cancels due to the mentioned coincidence among symplectic characters indexed by single row partitions. The prefactor also cancels for greater values of $L$, due to the identity
\begin{equation} \label{spcoinc}
    sp_{\lambda}(x_{1},\dots,x_{N}) = (-1)^{\lambda_{1}(\lambda_{1}+1)/2} sp_{\widetilde{\lambda}}(x_{1},\dots,x_{N}),
\end{equation}
where $\widetilde{\lambda}$ is the partition that results from rotating by 180º the complement of $\lambda$ in the rectangle\footnote{For instance, we have $sp_{(32222221)}(x_{1},x_{2},x_{3}) = sp_{(332111111)}(x_{1},x_{2},x_{3})$, with $N=M=3$. The second partition $(332111111)$ is obtained after rotating the complement of the first partition $(32222221)$ in the rectangle $(3^{10})$.} $(\lambda_{1}^{2N+\lambda_{1}+1})$. See appendix A for a proof of this identity.

\begin{center}
\begin{figure}
    \begin{scaletikzpicturetowidth}{.48\textwidth}
    \input{SN10L1tex.tex}
    \end{scaletikzpicturetowidth}
    \begin{scaletikzpicturetowidth}{.48\textwidth}
    \input{eON6L1tex.tex}
    \end{scaletikzpicturetowidth}
\caption{For each $n$ in the $x$ axis, the $y$ axis shows the coefficient of the monomial $q^{n/2}$ in $Z^{(L=1,m=0)}_{Sp(20)}$ (left) and $Z^{(L=1,m=0)}_{O(12)}$ (right).}
\label{f.SpO}
\end{figure}
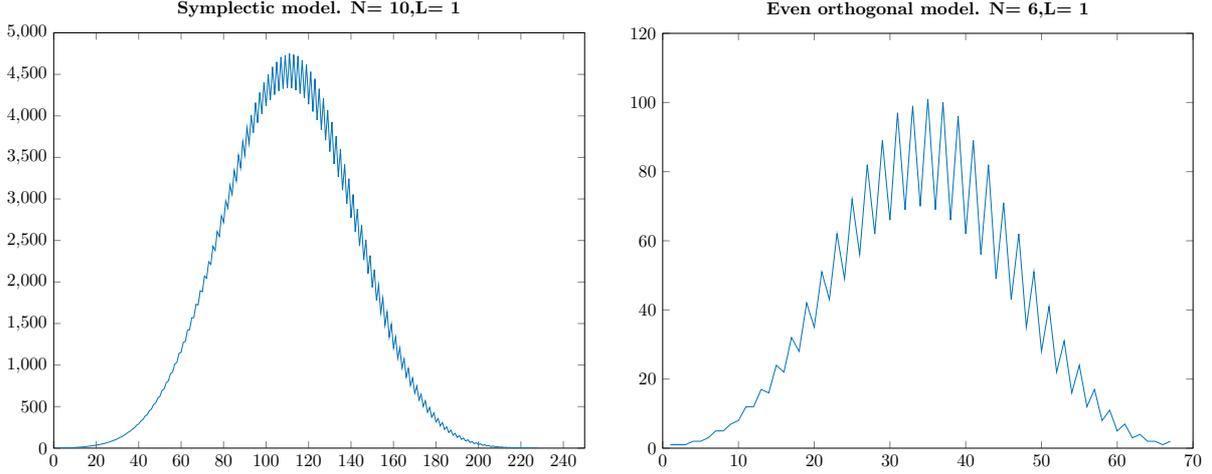
\end{center}

\subsubsection*{Large-$N$ limit}
Using identity \eqref{qdimG} and the dual Cauchy identity \eqref{cauchyun} we see that if $m\neq0$ we have
\begin{align}
    \lim_{N\rightarrow\infty}&Z^{(L,m)}_{Sp(2N)} = e^{Lm}(1-e^{-2m})^{-L(L+1)/2}\prod_{k=1}^{\infty}\frac{1}{(1-e^{-m}q^{k-1/2})^{L}}. \nonumber
\intertext{For the massless case, we can proceed as in the unitary model, and use known results on the asymptotics of Toeplitz$\pm$Hankel determinants generated by functions with Fisher-Hartwig singularities. It follows from \eqref{FHGN} that for a single singularity at $-1$ with parameters $\alpha=L$ and $\beta=0$ we have}
    &Z^{(L,m=0)}_{Sp(2N)} \sim \left(\frac{N}{2}\right)^{L(L+1)/2}\frac{\pi^{L/2}G(3/2)}{G(3/2+L)}\prod_{k=1}^{\infty}\frac{1}{(1-q^{k-1/2})^{L}} \label{FHsymp}
\end{align}
as $N\rightarrow\infty$. Table \ref{t.FH} shows some numerical tests of the accuracy of this formula.

\subsection{Orthogonal groups}

A similar reasoning applies to the orthogonal groups. For the even orthogonal group, it follows from \eqref{dcauchyeort} that
\begin{equation} \label{eoLm}
    Z^{(L,m)}_{SO(2N)} = e^{Lm}(1-e^{-2m})^{-L(L-1)/2}\sum_{\mu}e^{-|\mu|m}s_{\mu'}(1^{L})\langle W_{\mu} \rangle_{SO(2N)}.
\end{equation}
The even orthogonal characters verify $o_{\mu}^{even}(x_{1},\dots,x_{N})=0$ if $l(\mu)-\mu_{1}+1>2N$, and thus the sum above is now over all the partitions $\mu$ contained in the rectangle $(L^{2N+L-1})$ (a similar reasoning to the symplectic case holds, and in the end one can replace every even orthogonal Schur function $o_{\mu}^{even}(U)$ in the sum by the corresponding Wilson loop $\langle W_{\mu} \rangle_{SO(2N)}$). See figure \ref{f.SpO} for an example of this spectrum. A direct computation shows also that for $L=1$ the sum simplifies to
\begin{align*}
    &Z^{(L=1,m)}_{SO(2N)} = e^{m}\sum_{k=0}^{2N}e^{-km}q^{Nk-k^{2}/2}o_{(1^{k})}(1,q,\dots,q^{N-1}) = \\
    &=e^{m}\sum_{k=0}^{N-1}e^{-km}(1+e^{-(N-k)2m})q^{Nk-k^{2}/2}o_{(1^{k})}(1,q,\dots,q^{N-1}) + e^{-(N-1)m}q^{N^{2}/2}o_{(1^{N})}(1,q,\dots,q^{N-1}).
\end{align*}
As in the symplectic model, the prefactor $(1-e^{-2m})^{-L(L-1)/2}$ in \eqref{eoLm} cancels for higher values of $L$, due to the identity
\begin{equation} \label{oecoinc}
    o^{even}_{\lambda}(x_{1},\dots,x_{N}) = (-1)^{\lambda_{1}(\lambda_{1}-1)/2}o_{\widetilde{\lambda}}^{even}(x_{1},\dots,x_{N}),
\end{equation}
where $\widetilde{\lambda}$ is the partition obtained from rotating 180º the complement of $\lambda$ in the rectangular diagram $(\lambda_{1}^{2N+\lambda_{1}-1})$. See appendix A for a proof of identity \eqref{oecoinc}.

For the odd orthogonal group we have 
\begin{equation}
    Z^{(L,m)}_{SO(2N+1)} = e^{Lm}(1+e^{-m})^{-L}(1-e^{-2m})^{-L(L-1)/2}\sum_{\mu}e^{-|\mu|m}s_{\mu'}(1^{L})\langle W_{\mu} \rangle_{SO(2N+1)}, \label{ooLm}
\end{equation}
using \eqref{dcauchyoort}. Since $o_{\mu}^{odd}(x_{1},\dots,x_{N})=0$ whenever $l(\mu)-\mu_{1}>2N$, we see that the sum runs now over all the partitions $\mu$ contained in the rectangular shape $(L^{2N+L})$. The $L=1$ model can be computed explicitely, yielding 
\begin{align*}
    Z^{(L=1,m)}_{SO(2N+1)} &= e^{m}(1+e^{-m})^{-1}\sum_{k=0}^{2N+1}e^{-km}q^{Nk+k/2-k^{2}/2}o_{(1^{k})}^{odd}(q^{1/2},q^{3/2},\dots,q^{N-1/2}) = \\
    &=e^{m}(1+e^{-m})^{-1}\sum_{k=0}^{N}e^{-km}(1+e^{-(N-k+1/2)2m})q^{Nk+k/2-k^{2}/2}o_{(1^{k})}^{odd}(q^{1/2},\dots,q^{N-1/2}).
\end{align*}
As above, the prefactor $(1-e^{-2m})^{-L(L-1)/2}$ cancels for every $L$, in this time because of the identity
\begin{equation} \label{oocoinc}
    o^{odd}_{\lambda}(x_{1},\dots,x_{N}) = (-1)^{\lambda_{1}(\lambda_{1}-1)/2}o^{odd}_{\widetilde{\lambda}}(x_{1},\dots,x_{N}),
\end{equation}
where $\widetilde{\lambda}$ is the complement of the partition $\lambda$ in the rectangle $(\lambda_{1}^{2N+\lambda_{1}})$, rotated by 180º.

Using the dual Cauchy identities \eqref{dcauchyeort},\eqref{dcauchyoort} and identities \eqref{kauffeort} and \eqref{kauffoort} we see that also for the orthogonal models we have that
\begin{equation*}
    \lim_{q\rightarrow1}Z^{(L,m)}_{SO(2N)} = \lim_{q\rightarrow1}Z^{(L,m)}_{SO(2N+1)} = e^{Lm}(1+e^{-m})^{2NL},
\end{equation*}
preserving the $(N,L)$ duality and coincidence of the models in this limit.

\subsubsection*{Large-$N$ limit}
As in the symplectic model, using \eqref{qdimG} and the Cauchy identity \eqref{cauchyun} we see that if $m\neq 0$ then we have
\begin{align*}
    \lim_{N\rightarrow\infty}&Z^{(L,m)}_{SO(2N)} = e^{Lm}(1-e^{-2m})^{-L(L-1)/2}\prod_{k=1}^{\infty}\frac{1}{(1-e^{-m}q^{k-1/2})^{L}}
\intertext{and}
    \lim_{N\rightarrow\infty}&Z^{(L,m)}_{SO(2N+1)} = e^{Lm}(1+e^{-m})^{-L}(1-e^{-2m})^{-L(L-1)/2}\prod_{k=1}^{\infty}\frac{1}{(1-e^{-m}q^{k-1/2})^{L}}.
\end{align*}
If $m=0$ we can use again the known results on Fisher-Hartwig asymptotics reviewed in the appendix \eqref{FHGN} to obtain that, as $N\rightarrow\infty$,
\begin{align} \begin{split}
    &Z^{(L,m=0)}_{SO(2N)} \sim \left(\frac{N}{2}\right)^{L(L-1)/2}\frac{(4\pi)^{L/2}G(1/2)}{G(1/2+L)}\prod_{k=1}^{\infty}\frac{1}{(1-q^{k-1/2})^{L}}, \\
    &Z^{(L,m=0)}_{SO(2N+1)} \sim \left(\frac{N}{2}\right)^{L(L-1)/2}\frac{(\pi/4)^{L/2}G(1/2)}{G(1/2+L)}\prod_{k=1}^{\infty}\frac{1}{(1-q^{k-1/2})^{L}}. \end{split} \label{FHort}
\end{align}

\section{Conclusions and Outlook}

We summarize now the new and main results, since we have given considerable background mathematical results. We did so to make the work more self-contained and because the use of symmetric functions in random matrix computations of gauge theories is not so widespread. 

The results presented in Section 2 contain background material and some new results. While new, Theorems 1, 2 and 5 follow straightforwardly from known and classical results. Theorems 3, 4 and 6 are more involved. Whenever there is some partial overlap with the literature, as happens in Corollary 1 for example, detailed references (to the best of our knowledge) are provided.

Theorem $6$ is one of the main results, where we show that averages of the product of two characters are independent, in the large $N$ limit, of the Lie group $G$. We could relate this result, for the particular case of one character insertion -covered by our theorem- to the folklore result in gauge theory of the coincidence of Wilson loops averages in the large $N$ limit, regardless of the gauge group being the unitary, orthogonal or symplectic group. If we deal with a gauge theory that has a matrix model description (for example through localization), that is the content of our Theorem 6.

One instance of such theory is Chern-Simons theory on $S^3$, and we use in Section 3 some of the techniques explained to compute explicitly observables of that theory, up to Hopf links, for symplectic and orthogonal gauge groups.

As explained in the Introduction and in the last Section, some fermionic models with matrix degrees of freedom have been found to have a description (for the partition function) in terms of a characteristic polynomial insertion in a $U(N)$ Chern-Simons matrix model (corresponding to a Stieltjes-Wigert ensemble). We use the tools and results explained in this article to study such an average but considering instead matrix integration over $G(N)$. We obtain the same type of results as in the previous works: the Hilbert space picture of the partition function holds, in the sense that all the coefficients in the partition function are positive integers, and therefore interpreted as degeneracies of the quantum energy levels. 

We computed these degeneracies and energy levels using character expansion and plotted the corresponding distributions, but it would be interesting to study with more detail the statistical properties of such distribution.

\begin{acknowledgements}
We thank Jorge Lobera, Elia Bisi and David Pérez-García for valuable discussions and correspondence. The work of DGG was supported by the Fundação para a Ciência e a Tecnologia through the LisMath scholarship PD/BD/113627/2015. The work of MT was partially supported by the Fundação para a Ciência e a Tecnologia through its program Investigador FCT IF2014, under contract IF/01767/2014. The  work  is  also  partially supported by FCT Project PTDC/MAT-PUR/30234/2017.
\end{acknowledgements}

\section*{Appendix A: Characters of $G(N)$ and symmetric functions}

We summarize below some basic facts about partitions, the characters of the classical groups and symmetric functions, and list some of their properties. See \cite{MacDonald},\cite{KoikeTerada} for more details.

A partition $\lambda=(\lambda_{1},\dots,\lambda_{l})$ is a finite and non-increasing sequence of positive integers. The number of nonzero entries is called the length of the partition and is denoted by $l(\lambda)$, and the sum $|\lambda|=\lambda_{1}+\dots+\lambda_{l(\lambda)}$ is called the weight of the partition. The entry $\lambda_{j}$ is understood to be zero whenever the index $j$ is greater than the length of the partition. The notation $(a^{b})$ stands for the partition with exactly $b$ nonzero entries, all equal to $a$. A partition can be represented as a Young diagram, by placing $\lambda_{j}$ left-justified boxes in the $j$-th row of the diagram; the conjugate partition $\lambda^{\prime }$ is then obtained as the partition which diagram has as rows the columns of the diagram of $\lambda$.

Let $\lambda$ be a partition of length $l(\lambda)\leq N$. The characters associated to the irreducible representation indexed by $\lambda$ of each of the groups $G(N)$ are given by\footnote{Recall that the character \eqref{ooddchar} does not correspond to an irreducible representation of $O(2N)$ if $\lambda_{N}\neq0$. This fact is not relevant for our purposes so we ignore it throughout the paper and work with the algebraic expression \eqref{ooddchar}; minor modifications to the derivations allow a treatment of the general case.}
\begin{align}
    &\chi^{\lambda}_{U(N)}(U)=\frac{\det {M_{U(N)}^{\lambda }(z)}}{\det {M_{U(N)}(z)}}=\frac{\det {\left( z_{j}^{N-k+\lambda _{k}}\right) _{j,k=1}^{N}}}{\det {\left( z_{j}^{N-k}\right) _{j,k=1}^{N}}},  \label{unchar} \\
    &\chi^{\lambda}_{Sp(2N)}(U)=\frac{\det {M_{Sp(2N)}^{\lambda }(z)}}{\det {M_{Sp(2N)}(z)}}=\frac{\det {\left(z_{j}^{N-k+\lambda_{k}+1}-z_{j}^{-(N-k+\lambda_{k}+1)}\right)_{j,k=1}^{N}}}{\det {\left(z_{j}^{N-k+1}-z_{j}^{-(N-k+1)}\right) _{j,k=1}^{N}}}, \\
    &\chi^{\lambda}_{SO(2N)}(U)=\frac{\det {M_{SO(2N)}^{\lambda }(z)}}{\det {M_{SO(2N)}(z)}}=\frac{\det {\left(z_{j}^{N-k+\lambda_{k}}+z_{j}^{-(N-k+\lambda _{k})}\right) _{j,k=1}^{N}}}{\det {\left( z_{j}^{N-k}+z_{j}^{-(N-k)}\right) _{j,k=1}^{N}}}, \label{ooddchar} \\
    &\chi^{\lambda}_{SO(2N+1)}(U)=\frac{\det {M_{SO(2N+1)}^{\lambda }(z)}}{\det {M_{SO(2N+1)}(z)}}=\frac{\det {\left( z_{j}^{N-k+\lambda _{k}+\frac{1}{2}}-z_{j}^{-(N-k+\lambda _{k}+\frac{1}{2})}\right) _{j,k=1}^{N}}}{\det {\left(z_{j}^{N-k+\frac{1}{2}}-z_{j}^{-(N-k+\frac{1}{2})}\right) _{j,k=1}^{N}}}, \label{oevenchar}
\end{align}
where the $z_{j}=e^{i\theta_{j}}$ are the nontrivial eigenvalues of the matrices $U$. The determinants in the denominators above have the explicit evaluations \cite{Kratt}
\begin{align}
    &\det{M_{U(N)}(z)} = \det{\left(z_{j}^{N-k}\right)_{j,k=1}^{N}} = \prod_{1\leq j<k\leq N}(z_{j}-z_{k}),  \label{detun} \\
    &\det{M_{Sp(2N)}(z)} = \det{\left(z_{j}^{N-k+1}-z_{j}^{-(N-k+1)}\right)_{j,k=1}^{N}} = \prod_{1\leq j<k\leq N}(z_{j}-z_{k})(1-z_{j}z_{k})\prod_{j=1}^{N} (z_{j}^{2}-1)z_{j}^{-N}, \\
    &\det{M_{SO(2N)}(z)} = \det{\left(z_{j}^{N-k+\frac{1}{2}}-z_{j}^{-(N-k+\frac{1}{2})}\right)_{j,k=1}^{N}} = 2\prod_{1\leq j<k\leq N}(z_{j}-z_{k})(1-z_{j}z_{k})\prod_{j=1}^{N}z_{j}^{-N+1}, \\
    &\det{M_{SO(2N+1)}(z)} = \det{\left(z_{j}^{N-k}+z_{j}^{-(N-k)}\right)_{j,k=1}^{N}} = \prod_{1\leq j<k\leq N}(z_{j}-z_{k})(1-z_{j}z_{k})\prod_{j=1}^{N}(z_{j}-1)z_{j}^{-N+1/2}. \label{detone}
\end{align}
Given a (possibly infinite) set of variables $x=(x_{1},x_{2},\dots)$, the complete homogeneous and elementary symmetric polynomials are defined as
\begin{equation}  \label{ek}
    h_{k}(x) = \sum_{i_{1}\leq \dots \leq i_{k}}x_{i_{1}}\dots x_{i_{k}}, \qquad e_{k}(x) = \sum_{i_{1}< \dots < i_{k}}x_{i_{1}}\dots x_{i_{k}},
\end{equation}
respectively, for every positive integer $k$, together with the conditions $h_{0}=e_{0}=1$ and $h_{k}=e_{k}=0$ for negative integers $k$. Using these functions, one can define the Schur, symplectic Schur, and even/odd orthogonal Schur functionsby means of the Jacobi-Trudi identities 
\begin{align}
    s_{\lambda}(x) =& \det\left(h_{j-k+\lambda_{k}}(x)\right)_{j,k=1}^{l(\lambda)} = \det\left(e_{j-k+\lambda_{k}^{\prime}}(x)\right)_{j,k=1}^{\lambda_{1}}, \label{JTschur} \\
    sp_{\lambda}(x) =& \frac{1}{2}\det\left(h_{\lambda_{j}-j+k}(x,x^{-1})+h_{\lambda_{j}-j-k+2}(x,x^{-1})\right)_{j,k=1}^{l(\lambda)} \label{JTsymp} \\
    =& \det\left(e_{\lambda_{j}^{\prime}-j+k}(x,x^{-1})-e_{\lambda_{j}^{\prime}-j-k}(x,x^{-1})\right)_{j,k=1}^{\lambda_{1}}  \label{JTsympd} \\
    o^{even}_{\lambda}(x) =& \det\left(h_{\lambda_{j}-j+k}(x,x^{-1})-h_{\lambda_{j}-j-k}(x,x^{-1})\right)_{j,k=1}^{l(\lambda)} \label{JTeort} \\
    =& \frac{1}{2}\det\left(e_{\lambda_{j}^{\prime}-j+k}(x,x^{-1})+e_{\lambda_{j}^{\prime }-j-k+2}(x,x^{-1})\right)_{j,k=1}^{\lambda_{1}}, \label{JTeortd} \\
    o^{odd}_{\lambda}(x) =& \det\left(h_{\lambda_{j}-j+k}(x,x^{-1},1)-h_{\lambda_{j}-j-k}(x,x^{-1},1)\right)_{j,k=1}^{l(\lambda)} \label{JToort} \\
    =& \frac{1}{2}\det\left(e_{\lambda_{j}^{\prime}-j+k}(x,x^{-1},1)+e_{\lambda_{j}^{\prime }-j-k+2}(x,x^{-1},1)\right)_{j,k=1}^{\lambda_{1}}. \label{JToortd}
\end{align}
They satisfy the Cauchy identities 
\begin{align}
    &\sum_{\nu}s_{\nu}(x)s_{\nu}(y) = \prod_{i,j=1}^{\infty}\frac{1}{1-x_{i}y_{j}},  \label{cauchyun} \\
    &\sum_{\nu}sp_{\nu}(x)s_{\nu}(y) =\prod_{i<j}(1-y_{i}y_{j})\prod_{i,j=1}^{\infty}\frac{1}{1-x_{i}y_{j}}\frac{1}{1-x_{i}^{-1}y_{j}},  \label{cauchysymp} \\
    &\sum_{\nu}o_{\nu}^{even}(x)s_{\nu}(y) = \prod_{i\leq j}(1-y_{i}y_{j})\prod_{i,j=1}^{\infty}\frac{1}{1-x_{i}y_{j}}\frac{1}{1-x_{i}^{-1}y_{j}},  \label{cauchyeort} \\
    &\sum_{\nu}o^{odd}_{\nu}(x)s_{\nu}(y) = \prod_{i\leq j}(1-y_{i}y_{j})\prod_{i,j=1}^{\infty}\frac{1}{1-x_{i}y_{j}}\frac{1}{1-x_{i}^{-1}y_{j}}\prod_{j=1}^{\infty}\frac{1}{1-y_{j}},  \label{cauchyoort}
\intertext{and dual Cauchy identities}
    &\sum_{\nu}s_{\nu}(x)s_{\nu^{\prime }}(y) = \prod_{i,j=1}^{\infty}(1+x_{i}y_{j}),  \label{dcauchyun} \\
    &\sum_{\nu}sp_{\nu}(x)s_{\nu^{\prime }}(y) = \prod_{i\leq j}(1-y_{i}y_{j})\prod_{i,j=1}^{\infty}(1+x_{i}y_{j})(1+x_{i}^{-1}y_{j}), \label{dcauchysymp} \\
    &\sum_{\nu}o_{\nu}^{even}(x)s_{\nu^{\prime }}(y) = \prod_{i<j}(1-y_{i}y_{j})\prod_{i,j=1}^{\infty}(1+x_{i}y_{j})(1+x_{i}^{-1}y_{j}).  \label{dcauchyeort} \\
    &\sum_{\nu}o^{odd}_{\nu}(x)s_{\nu^{\prime }}(y) = \prod_{i<j}(1-y_{i}y_{j})\prod_{i,j=1}^{\infty}(1+x_{i}y_{j})(1+x_{i}^{-1}y_{j})\prod_{j=1}^{\infty}(1+y_{j}). \label{dcauchyoort}
\end{align}

Since the groups $Sp(2N),O(2N),O(2N+1)$ can be embedded on the unitary group $U(2N)$ or $U(2N+1)$, the irreducible characters on each of these groups can be expressed in terms of the others, after applying the specialization homomorphisms $(z_{1},\dots,z_{2N})\mapsto (z_{1},\dots,z_{N},z_{1}^{-1},\dots,z_{N}^{-1})$ (for $Sp(2N)$,$O(2N)$) or $(z_{1},\dots,z_{2N+1})\mapsto (z_{1},\dots,z_{N},z_{1}^{-1},\dots,z_{N}^{-1},1)$ (for $O(2N+1)$). When seen as universal characters in the the ring symmetric functions, they have the following expansions \cite{KoikeTerada} 
\begin{align}
    &s_{\lambda}(x,x^{-1}) = \sum_{\alpha}\sum_{\beta^{\prime}\,even}c^{\lambda}_{\alpha\beta}sp_{\alpha}(x),  \label{restschursp} \\
    &s_{\lambda}(x,x^{-1}) = \sum_{\alpha}\sum_{\beta\,even}c^{\lambda}_{\alpha\beta}o^{even}_{\alpha}(x), \label{restschureort} \\
    &s_{\lambda}(x,x^{-1},1) = \sum_{\alpha}\sum_{\beta\,even}c^{\lambda}_{\alpha\beta}o^{odd}_{\alpha}(x), \label{restschuroort}
\end{align}
where $c^{\lambda}_{\alpha,\beta}$ are Littlewood-Richardson coefficients and we say that a partition is even if it has only even parts. Reciprocally,
\begin{align}
    &sp_{\lambda}(x) = \sum_{\alpha}\sum_{\beta\in T(N)}(-1)^{|\beta|/2}c^{\lambda}_{\alpha\beta}s_{\alpha}(x,x^{-1}) = \sum_{\beta\in T(N)} (-1)^{|\beta|/2}s_{\lambda/\beta}(x,x^{-1}), \label{restspschur} \\
    &o^{even}_{\lambda}(x) = \sum_{\alpha}\sum_{\beta\in R(N)}(-1)^{|\beta|/2}c^{\lambda}_{\alpha\beta}s_{\alpha}(x,x^{-1}) = \sum_{\beta\in R(N)} (-1)^{|\beta|/2}s_{\lambda/\beta}(x,x^{-1}) \label{resteortschur} \\
    &o^{odd}_{\lambda}(x) = \sum_{\alpha}\sum_{\beta\in R(N)}(-1)^{|\beta|/2}c^{\lambda}_{\alpha\beta}s_{\alpha}(x,x^{-1},1) = \sum_{\beta\in R(N)} (-1)^{|\beta|/2}s_{\lambda/\beta}(x,x^{-1},1) \label{restoortschur}
\end{align}
where $T(N)$ and $R(N)$ are the sets defined before theorem \ref{th.hopflinks}.

Let us record here a proof of identities \eqref{spcoinc},\eqref{oecoinc} and \eqref{oocoinc}, as we have been unable to find them in the literature.

\begin{theorem}
Let $\lambda=(1^{a_{1}}2^{a_{2}}\dots M^{a_{M}})$ be a partition, written in frequency notation. That is, $\lambda$ is the partition with exactly $a_{M}$ parts equal to $M$, $a_{M-1}$ parts equal to $M-1$, and so on. We have
\begin{equation*}
    sp_{\lambda}(x_{1},x_{2},\dots,x_{N}) = (-1)^{M(M+1)/2}sp_{\widetilde{\lambda}}(x_{1},x_{2},\dots,x_{N}),
\end{equation*}
where $\widetilde{\lambda}=(1^{a_{M-1}}2^{a_{M-2}}\dots (M-2)^{a_{2}}(M-1)^{a_{1}}M^{2N+M+1-a_{1}-a_{2}-\dots-a_{M}})$ is the partition that results from rotating 180º the complement of $\lambda$ in the rectangular diagram $(M^{2N+M+1})$.
\end{theorem}
\begin{proof}
First of all, note that for $\lambda$ as above we have
\begin{equation*}
    \lambda'=(a_{M}+a_{M-1}+\dots+a_{1},a_{M}+a_{M-1}+\dots+a_{2},\dots,a_{M}+a_{M-1},a_{M}),
\end{equation*}
using the standard notation for partitions. Let us denote the $j$-th entry of $\lambda'$ by $b_{j}$ to simplify the exposition. It follows from the Jacobi-Trudi identity \eqref{JTsympd} and \eqref{elemid} that
\begin{align*}
    &(-1)^{M}sp_{\lambda} = (-1)^{M}\begin{vmatrix} e_{b_{1}}-e_{b_{1}-2} & e_{b_{1}+1}-e_{b_{1}-3} & \dots & e_{b_{1}+M-1}-e_{b_{1}-M-1} \\
    e_{b_{2}-1}-e_{b_{2}-3} & e_{b_{2}}-e_{b_{2}-4} & \dots & e_{b_{2}+M-2}-e_{b_{2}-M-2} \\
    \vdots & \vdots & & \vdots \\
    e_{b_{M-1}-M+2}-e_{b_{M-1}-M} & e_{b_{M-1}-M+3}-e_{b_{M-1}-M-1} & \dots & e_{b_{M-1}+1}-e_{b_{M-1}-2M+1} \\
    e_{b_{M}-M+1}-e_{b_{M}-M-1} & e_{b_{M}-M+2}-e_{b_{M}-M-2} & \dots & e_{b_{M}}-e_{b_{M}-2M}
    \end{vmatrix} = \nonumber \\
    &(-1)^{M}\begin{vmatrix} sp_{(1^{b_{1}})} & sp_{(1^{b_{1}+1})}+sp_{(1^{b_{1}-1})} & \dots & sp_{(1^{b_{1}+M-1})}+\dots+sp_{(1^{b_{1}-M+1})} \\
    sp_{(1^{b_{2}-1})} & sp_{(1^{b_{2}})}+sp_{(1^{b_{2}-2})} & \dots & sp_{(1^{b_{2}+M-2})}+\dots+sp_{(1^{b_{2}-M})} \\
    \vdots & \vdots & & \vdots \\
    sp_{(1^{b_{M-1}-M+2})} & sp_{(1^{b_{M-1}-M+3})}+sp_{(1^{b_{M-1}-M+1})} & \dots & sp_{(1^{b_{M-1}+1})}+\dots+sp_{(1^{b_{M-1}-2M+3})} \\
    sp_{(1^{b_{M}-M+1})} & sp_{(1^{b_{M}-M+2})}+sp_{(1^{b_{M}-M})} & \dots & sp_{(1^{b_{M}})}+\dots+sp_{(1^{b_{M}-2M+2})}
    \end{vmatrix} = \nonumber \\
    &\begin{vmatrix} sp_{(1^{2N+2-b_{1}})} & sp_{(1^{2N+1-b_{1}})}+sp_{(1^{2N+3-b_{1}})} & \dots & sp_{(1^{2N+3-M-b_{1}})}+\dots+sp_{(1^{2N+1+M-b_{1}})} \\
    sp_{(1^{2N+3-b_{2}})} & sp_{(1^{2N+2-b_{2}})}+sp_{(1^{2N+4-b_{2}})} & \dots & sp_{(1^{2N+4-M-b_{2}})}+\dots+sp_{(1^{2N+2+M-b_{2}})} \\
    \vdots & \vdots & & \vdots \\
    sp_{(1^{2N+M-b_{M-1}})} & sp_{(1^{2N-1+M-b_{M-1}})}+sp_{(1^{2N+1+M-b_{M-1}})} & \dots & sp_{(1^{2N+1-b_{M-1}})}+\dots+sp_{(1^{2N-1+2M-b_{M-1}})} \\
    sp_{(1^{2N+1+M-b_{M}})} & sp_{(1^{2N+M-b_{M}})}+sp_{(1^{2N+2+M-b_{M}})} & \dots & sp_{(1^{2N+2-b_{M}})}+\dots+sp_{(1^{2N+2M-b_{M}})}
    \end{vmatrix} \nonumber
\end{align*}
(we have omitted the dependence on $x$ for ease of notation). Reversing the order of the rows of the last determinant above, we see that it corresponds to another symplectic Schur function indexed by some partition $\mu$. Comparing this with the second determinant above we see that $\mu$ verifies
\begin{align*}
    \mu_{1}' &= 2N+M+1-b_{M} = 2N+M+1-a_{M}, \\
    \mu_{2}' &= 2N+M+1-b_{M-1} = 2N+M+1-a_{M}-a_{M-1}, \\
    &\vdots \\
    \mu_{M}' &= 2N+M+1-b_{1} = 2N+M+1-a_{M}-a_{M}-a_{M-1}-\dots-a_{1},
\end{align*}
proving the desired result.
\end{proof}
The proof of identities \eqref{oecoinc} and \eqref{oocoinc} follows analogously, using the corresponding Jacobi-Trudi identities.


\section*{Appendix B: Large-$N$ limit of Toeplitz and Toeplitz$\pm$Hankel
determinants}

The classical strong Szeg\H o limit theorem describes the large-$N$ behaviour of Toeplitz determinants generated by sufficiently smooth functions. We record below its statement and a generalization for the determinants of Toeplitz$\pm$Hankel matrices due to Johansson \cite{Joh} (see also \cite{BasorEhrhardt,BasEhr2017}).

\begin{theorem*}[Szeg\H o, Johansson]
Let $f(e^{i\theta})=\exp(\sum_{k=1}^{\infty}V_{k}e^{ik\theta})$, with $%
\sum_{k}|V_{k}|<\infty$ and $\sum_{k}k|V_{k}|^{2}<\infty$, and define $f(U)$
by formula \eqref{fprod}, where $U$ belongs to any of the groups $G(N)$. We
have 
\begin{align}
&\lim_{N\rightarrow\infty}\int_{U(N)}f(U)dU
=\exp\left(\sum_{k=1}^{\infty}kV_{k}^{2}\right).  \label{szego} \\
&\lim_{N\rightarrow\infty}\int_{Sp(2N)}f(U)dU = \exp\left(\frac{1}{2}%
\sum_{k=1}^{\infty}kV_{k}^{2}-\sum_{k=1}^{\infty}V_{2k} \right),
\label{szegosymp} \\
&\lim_{N\rightarrow\infty}\int_{SO(2N)}f(U)dU = \exp\left(\frac{1}{2}%
\sum_{k=1}^{\infty}kV_{k}^{2}+\sum_{k=1}^{\infty}V_{2k} \right),
\label{szegoeort} \\
&\lim_{N\rightarrow\infty}\int_{SO(2N+1)}f(U)dU = \exp\left(\frac{1}{2}%
\sum_{k=1}^{\infty}kV_{k}^{2}-\sum_{k=1}^{\infty}V_{2k-1} \right).
\label{szegooort}
\end{align}
\end{theorem*}

We have stated the theorem for slightly different integrals that those
appearing in \cite{Joh}; the result follows after using the mapping $\cos\theta_{j}\mapsto x_{j}$ in the integrals\footnote{The relation is more apparent working directly with the ``trigonometric" expression of Haar measure on $G(N)$, see for instance equations (3.3)-(3.5) in \cite{DCMOC}.} \eqref{intGN}. This allows to
express the integrals in terms of the orthogonal polynomials with respect to
a modified weight on $[-1,1]$, which relation with the orthogonal
polynomials with respect to the original weight is well known \cite{Szego}
(see also \cite{BaikRains}).

The asymptotic behaviour of Toeplitz determinants generated by functions
that do not satisfy the hypotheses in Szeg\H o's theorem has attracted a lot
of interest over the years \cite{DIKrev}. Such functions are typically
studied in terms of their factorization as a sufficiently smooth function
(in the sense of Szeg\H o's theorem) and a finite number of so-called
Fisher-Hartwig singularities \cite{BS}
\begin{equation}
    \varphi_{z,\alpha,\beta}(ze^{i\theta}) = |1-e^{i\theta}|^{2\alpha}e^{i\beta(\theta-\pi)} = (1-e^{i\theta})^{\alpha+\beta}(1-e^{-i\theta})^{\alpha-\beta}, \label{pureFH}
\end{equation}
where $z$ is a point on the unit circle, $\text{Re}(\alpha_{r})>-1/2$ and $\beta_{r}\in\mathbb{C}$. This function may have a zero, a pole, or an oscillatory singularity at $z$, depending on the value of $\alpha$, and a jump at the same point if $\beta$ is not an integer.

For our purposes, we only need to consider Toeplitz determinants generated by functions with a single Fisher-Hartwig singularity. This fact, together with the definition \eqref{fprod} allow us to consider only particular examples of the very general results known for this kind of asymptotics. What follows is a particular case of a theorem of Widom \cite{Widom} for functions of the form \eqref{fprod} with a single singularity, adapted for this setting. See \cite{Ehrhardt,DIK} for more general results on the topic.

\begin{theorem*}[Widom]
Let $f$ be given by 
\begin{equation}
f(e^{i\theta}) = e^{V(e^{i\theta})}(1-e^{i(\theta-\theta_{0})})^{\alpha},
\label{f+FH}
\end{equation}
where $\text{Re}(\alpha)>-1/2$, $0<\theta_{0}<2\pi$, and the
potential $V(e^{i\theta})=\sum_{k=1}^{\infty}V_{k}e^{ik\theta}$ satisfies $%
\sum_{k}|V_{k}|<\infty$ and $\sum_{k}k|V_{k}|^{2}<\infty$, as in Szeg\H o's
theorem. Define $f(U)$ by \eqref{fprod} for any $U\in U(N)$. Then, as $%
N\rightarrow\infty$, we have 
\begin{equation}
\int_{U(N)}f(U)dU \sim \exp{\left(\sum_{k=1}^{\infty}kV_{k}^{2}\right)}
N^{\alpha^{2}} e^{-2\alpha V(e^{i\theta_{0}})}\frac{G^{2}(\alpha+1)}{G(2\alpha+1)}.
\label{FHun}
\end{equation}
\end{theorem*}

The asymptotic behaviour of Toeplitz$\pm$Hankel determinants generated by
functions with Fisher-Hartwig singularities has also been studied. As above, we state only a particular case of a theorem of Deift, Its and Krasovsky \cite{DIK} for functions with a single singularity at the point $z=-1$, which will be enough for our purposes. See \cite{DIK} for general results on Fisher-Hartwig asymptotics of Toeplitz$\pm$Hankel determinants.

\begin{theorem*}[Deift, Its, Krasovsky]
Let $f$ be given by \eqref{f+FH}, with $\theta_{0}=\pi$, and define $f(U)$ by \eqref{fprod} for any $U\in Sp(2N),SO(2N),SO(2N+1)$. Then, as $N\rightarrow\infty$, we have 
\begin{equation}
\int_{G(N)}f(U)dU \sim \left(\int_{G(N)}e^{V(U)}dU\right) e^{-\alpha
V(-1)}N^{\alpha^{2}/2+\alpha t}2^{-\alpha^{2}/2-\alpha(s+t-1/2)}\frac{%
\pi^{\alpha/2}G(t+1)}{G(\alpha+t+1)},  \label{FHGN}
\end{equation}
where $s$ and $t$ depend on the group $G(N)$ and are given by 
\begin{align*}
Sp(2N): s=t=\frac{1}{2},\qquad SO(2N): s=t=-\frac{1}{2},\qquad SO(2N+1): s=%
\frac{1}{2},t=-\frac{1}{2}.
\end{align*}
\end{theorem*}

Recall that the factor $\lim_{N}\int_{G(N)}e^{V(U)}dU$ in \eqref{FHGN} can be computed by means of the Szeg\H o-Johansson theorem above.

\section*{Appendix C: $S$ and $T$ matrices}

The $S$ and $T$ matrices are central in the study of modular tensor categories, which
has its origins in the study of rational conformal field
theories \cite{Moore:1988uz}, and also underlies a topological quantum field theory in 3-dimensions. 

The modular group $SL(2,\mathbb{Z})$ is the most basic example of a discrete nonabelian group. Two particular
elements in $SL(2,\mathbb{Z})$ are $S=
\begin{pmatrix}
0 & -1 \\ 
1 & 0 \end{pmatrix}$ and $T=\begin{pmatrix}
1 & 1 \\ 
0 & 1\end{pmatrix}$. It can then be proven that the matrices $S$ and $T$ generate $SL(2,\mathbb{Z})$.

These modular $T$ and $S$ matrices, are generated, respectively, by a Dehn  twist and a $90$º rotation on the torus. Recall that a Dehn twist essentially consists in cutting 
up a torus along one axis, twisting the edge by $360$º and glueing the two edges back.

To have some idea of the relationship with topological QFT one can recall that there are a finite set of objects associated with a two dimensional surface and 
that the topological nature of the association means that the mapping class group of the surface acts on these objects. A
known example originates in the $G/G$ WZW theory on $T^{2}$ for $G=SU(N)$.
In this case, the objects are the conformal blocks of the theory, which are
the characters of $\widehat{SU(N)}_{k}$, the affine Lie algebra of $SU(N)$
at level $k$. The action of the modular group on the characters of $\widehat{%
SU(N)}_{k}$ is given by 
\begin{equation}
S_{\lambda \,\mu }=\sum_{w\in W}(-1)^{|w|}\,\mbox{exp}\Big(\frac{i\pi }{k+N}%
(\lambda +\rho ,w(\mu +\rho ))\Big)  \label{xx} \\
=s_{\lambda }(q^{\frac{1}{2}},q^{\frac{3}{2}},\mathellipsis,q^{N-\frac{1}{2}%
})\,s_{\mu }(q^{\frac{1}{2}-\lambda _{1}},q^{\frac{3}{2}-\lambda _{2}},%
\mathellipsis,q^{N-\frac{1}{2}-\lambda _{N}}),
\end{equation}%
where $s_{\lambda }(x_{1},\mathellipsis,x_{N})$ is a Schur polynomial.
Because the WZW theory on $T^{2}$ is related to the canonical quantization
of the Chern-Simons theory on $T^{2}\times \mathbb{R}$ the space of
conformal blocks of the WZW theory on $T^{2}$ is also the Hilbert space of
the Chern-Simons theory, where the normalized S-matrix $S_{\lambda \,\mu
}/S_{\emptyset \,\emptyset }$ is the Hopf link invariant. The matrix model results 
are in the so-called Seifert framing instead of the canonical framing of the three manifold. 
Starting from $S^{2}\times S^{1}$ one generates $S^3$, by action of $T^mST^n$. While the 
canonical framing for $S^3$ corresponds to $m=n=0$, one obtains a $U(1)$-invariant Seifert framing for $n+m=2$ \cite{Blau:2006gh}.

The $T$ and $S$ matrices encode the information of quasi-particles non-Abelian statistics and their fusion and are central
in the description of topological order \cite{top}. Remarkably, such different braiding statistics, described by the matrices, can also be extracted in many models using wavefunction overlaps \cite{top}. In this regard, for example, the minor description of the modular matrix elements and its associated integral representation of random matrix type studied here is conductive to interpretation in term of quantum amplitudes, of the Loschmidt echo type, of certain $1d$ spin chain \cite{Perez-Garcia:2014aba}.

\end{document}

%% file: U2N16L1tex.tex
%
%
\definecolor{mycolor1}{rgb}{0.00000,0.44700,0.74100}%
\begin{tikzpicture}[scale=\tikzscale]

\begin{axis}[%
width=4.521in,
height=3.566in,
at={(0.758in,0.481in)},
scale only axis,
xmin=0,
xmax=300,
ymin=0,
ymax=60000000,
axis background/.style={fill=white},
title style={font=\bfseries},
title={N= 16,L= 1}
]
\addplot [color=mycolor1, forget plot]
  table[row sep=crcr]{%
1	17\\
2	32\\
3	45\\
4	88\\
5	125\\
6	188\\
7	286\\
8	400\\
9	554\\
10	768\\
11	1047\\
12	1388\\
13	1844\\
14	2404\\
15	3109\\
16	4008\\
17	5094\\
18	6424\\
19	8074\\
20	10052\\
21	12442\\
22	15344\\
23	18778\\
24	22868\\
25	27737\\
26	33464\\
27	40184\\
28	48084\\
29	57274\\
30	67940\\
31	80331\\
32	94600\\
33	110998\\
34	129848\\
35	151369\\
36	175892\\
37	203812\\
38	235416\\
39	271121\\
40	311432\\
41	356715\\
42	407488\\
43	464333\\
44	527736\\
45	598311\\
46	676760\\
47	763646\\
48	859680\\
49	965707\\
50	1082360\\
51	1210487\\
52	1350976\\
53	1504606\\
54	1672260\\
55	1854975\\
56	2053568\\
57	2269005\\
58	2502364\\
59	2754518\\
60	3026488\\
61	3319378\\
62	3634072\\
63	3971570\\
64	4332996\\
65	4719173\\
66	5131080\\
67	5569741\\
68	6035908\\
69	6530446\\
70	7054276\\
71	7607982\\
72	8192264\\
73	8807827\\
74	9455084\\
75	10134448\\
76	10846436\\
77	11591100\\
78	12368616\\
79	13179133\\
80	14022400\\
81	14898198\\
82	15806312\\
83	16746090\\
84	17716904\\
85	18718124\\
86	19748644\\
87	20807418\\
88	21893304\\
89	23004814\\
90	24140368\\
91	25298459\\
92	26477036\\
93	27674128\\
94	28887728\\
95	30115425\\
96	31354768\\
97	32603376\\
98	33858412\\
99	35117116\\
100	36376724\\
101	37634123\\
102	38886228\\
103	40130067\\
104	41362288\\
105	42579647\\
106	43779032\\
107	44956971\\
108	46110192\\
109	47235530\\
110	48329616\\
111	49389177\\
112	50411248\\
113	51392539\\
114	52330100\\
115	53221109\\
116	54062700\\
117	54852148\\
118	55587132\\
119	56265118\\
120	56883944\\
121	57441724\\
122	57936496\\
123	58366646\\
124	58730896\\
125	59027940\\
126	59256780\\
127	59416814\\
128	59507392\\
129	59528258\\
130	59479460\\
131	59361104\\
132	59173580\\
133	58917701\\
134	58594204\\
135	58204233\\
136	57749188\\
137	57230558\\
138	56649996\\
139	56009621\\
140	55311396\\
141	54557633\\
142	53750844\\
143	52893552\\
144	51988456\\
145	51038510\\
146	50046564\\
147	49015634\\
148	47948952\\
149	46849576\\
150	45720744\\
151	44565754\\
152	43387832\\
153	42190180\\
154	40976192\\
155	39748954\\
156	38511676\\
157	37267524\\
158	36019528\\
159	34770617\\
160	33523792\\
161	32281747\\
162	31047176\\
163	29822701\\
164	28610720\\
165	27413510\\
166	26233336\\
167	25072163\\
168	23931868\\
169	22814259\\
170	21720876\\
171	20653167\\
172	19612448\\
173	18599847\\
174	17616316\\
175	16662796\\
176	15739904\\
177	14848232\\
178	13988228\\
179	13160189\\
180	12364236\\
181	11600508\\
182	10868876\\
183	10169191\\
184	9501204\\
185	8864534\\
186	8258716\\
187	7683262\\
188	7137544\\
189	6620876\\
190	6132580\\
191	5671842\\
192	5237836\\
193	4829722\\
194	4446608\\
195	4087541\\
196	3751632\\
197	3437896\\
198	3145380\\
199	2873124\\
200	2620180\\
201	2385555\\
202	2168344\\
203	1967586\\
204	1782372\\
205	1611818\\
206	1455044\\
207	1311181\\
208	1179440\\
209	1059009\\
210	949116\\
211	849050\\
212	758092\\
213	675578\\
214	600868\\
215	533370\\
216	472492\\
217	417717\\
218	368520\\
219	324427\\
220	284996\\
221	249818\\
222	218484\\
223	190651\\
224	165976\\
225	144155\\
226	124900\\
227	107950\\
228	93064\\
229	80027\\
230	68636\\
231	58705\\
232	50076\\
233	42591\\
234	36124\\
235	30545\\
236	25752\\
237	21640\\
238	18124\\
239	15133\\
240	12592\\
241	10438\\
242	8624\\
243	7095\\
244	5816\\
245	4746\\
246	3860\\
247	3124\\
248	2516\\
249	2017\\
250	1608\\
251	1278\\
252	1008\\
253	790\\
254	616\\
255	478\\
256	368\\
257	282\\
258	212\\
259	161\\
260	120\\
261	88\\
262	64\\
263	47\\
264	32\\
265	22\\
266	16\\
267	10\\
268	8\\
269	5\\
270	4\\
271	1\\
272	1\\
};
\end{axis}
\end{tikzpicture}%

%% file: U2N6L2tex.tex
%
%
\definecolor{mycolor1}{rgb}{0.00000,0.44700,0.74100}%
\begin{tikzpicture}[scale=\tikzscale]

\begin{axis}[%
width=4.521in,
height=3.566in,
at={(0.758in,0.481in)},
scale only axis,
xmin=0,
xmax=100,
ymin=0,
ymax=600000,
axis background/.style={fill=white},
title style={font=\bfseries},
title={N=6, L=2}
]
\addplot [color=mycolor1, forget plot]
  table[row sep=crcr]{%
1	336\\
2	1008\\
3	1815\\
4	3480\\
5	5651\\
6	8464\\
7	12729\\
8	17824\\
9	24270\\
10	32448\\
11	42419\\
12	54216\\
13	68018\\
14	83816\\
15	101744\\
16	121936\\
17	143821\\
18	167960\\
19	193451\\
20	220672\\
21	249136\\
22	278296\\
23	308066\\
24	337776\\
25	367409\\
26	395816\\
27	423386\\
28	448768\\
29	472607\\
30	493904\\
31	512073\\
32	527408\\
33	539269\\
34	547968\\
35	552550\\
36	554104\\
37	551446\\
38	545848\\
39	536738\\
40	524480\\
41	509416\\
42	491464\\
43	471886\\
44	449720\\
45	426544\\
46	401728\\
47	376464\\
48	350544\\
49	324590\\
50	298688\\
51	273134\\
52	248536\\
53	224240\\
54	201512\\
55	179564\\
56	159456\\
57	140382\\
58	122888\\
59	106912\\
60	92216\\
61	79214\\
62	67376\\
63	57078\\
64	47760\\
65	39970\\
66	33008\\
67	27123\\
68	22064\\
69	17791\\
70	14360\\
71	11325\\
72	8992\\
73	6978\\
74	5432\\
75	4127\\
76	3136\\
77	2350\\
78	1752\\
79	1312\\
80	928\\
81	685\\
82	456\\
83	343\\
84	232\\
85	152\\
86	96\\
87	58\\
88	48\\
89	25\\
90	16\\
91	6\\
92	8\\
93	3\\
94	0\\
95	1\\
96	0\\
97	1\\
};
\end{axis}
\end{tikzpicture}%

%% file: SN10L1tex.tex
%
%
\definecolor{mycolor1}{rgb}{0.00000,0.44700,0.74100}%
\begin{tikzpicture}[scale=\tikzscale]

\begin{axis}[%
width=4.521in,
height=3.566in,
at={(0.758in,0.481in)},
scale only axis,
xmin=0,
xmax=250,
ymin=0,
ymax=5000,
axis background/.style={fill=white},
title style={font=\bfseries},
title={Symplectic model. N= 10,L= 1}
]
\addplot [color=mycolor1, forget plot]
  table[row sep=crcr]{%
1	1\\
2	1\\
3	1\\
4	2\\
5	2\\
6	3\\
7	4\\
8	5\\
9	6\\
10	8\\
11	10\\
12	11\\
13	13\\
14	15\\
15	18\\
16	21\\
17	24\\
18	27\\
19	32\\
20	35\\
21	40\\
22	45\\
23	52\\
24	57\\
25	65\\
26	72\\
27	82\\
28	90\\
29	102\\
30	111\\
31	127\\
32	137\\
33	155\\
34	168\\
35	190\\
36	203\\
37	228\\
38	244\\
39	275\\
40	291\\
41	326\\
42	344\\
43	385\\
44	403\\
45	449\\
46	470\\
47	524\\
48	543\\
49	602\\
50	624\\
51	694\\
52	713\\
53	789\\
54	810\\
55	898\\
56	916\\
57	1012\\
58	1031\\
59	1141\\
60	1154\\
61	1272\\
62	1285\\
63	1420\\
64	1425\\
65	1569\\
66	1571\\
67	1732\\
68	1724\\
69	1895\\
70	1883\\
71	2072\\
72	2045\\
73	2244\\
74	2211\\
75	2430\\
76	2380\\
77	2608\\
78	2549\\
79	2798\\
80	2719\\
81	2978\\
82	2889\\
83	3168\\
84	3054\\
85	3342\\
86	3216\\
87	3526\\
88	3374\\
89	3691\\
90	3524\\
91	3861\\
92	3666\\
93	4008\\
94	3798\\
95	4159\\
96	3917\\
97	4281\\
98	4024\\
99	4405\\
100	4117\\
101	4498\\
102	4193\\
103	4590\\
104	4255\\
105	4650\\
106	4301\\
107	4708\\
108	4328\\
109	4729\\
110	4338\\
111	4752\\
112	4333\\
113	4737\\
114	4309\\
115	4721\\
116	4269\\
117	4670\\
118	4213\\
119	4619\\
120	4140\\
121	4531\\
122	4052\\
123	4446\\
124	3950\\
125	4326\\
126	3834\\
127	4210\\
128	3707\\
129	4065\\
130	3570\\
131	3925\\
132	3422\\
133	3755\\
134	3267\\
135	3598\\
136	3107\\
137	3415\\
138	2942\\
139	3244\\
140	2774\\
141	3054\\
142	2605\\
143	2878\\
144	2435\\
145	2684\\
146	2266\\
147	2509\\
148	2100\\
149	2319\\
150	1936\\
151	2149\\
152	1778\\
153	1969\\
154	1626\\
155	1808\\
156	1478\\
157	1639\\
158	1337\\
159	1493\\
160	1205\\
161	1341\\
162	1080\\
163	1209\\
164	963\\
165	1075\\
166	855\\
167	961\\
168	754\\
169	844\\
170	662\\
171	748\\
172	578\\
173	649\\
174	501\\
175	569\\
176	433\\
177	489\\
178	372\\
179	424\\
180	316\\
181	358\\
182	267\\
183	308\\
184	225\\
185	257\\
186	188\\
187	218\\
188	156\\
189	179\\
190	129\\
191	151\\
192	105\\
193	121\\
194	85\\
195	101\\
196	68\\
197	79\\
198	54\\
199	65\\
200	43\\
201	51\\
202	34\\
203	41\\
204	26\\
205	30\\
206	19\\
207	25\\
208	15\\
209	18\\
210	11\\
211	14\\
212	8\\
213	10\\
214	6\\
215	8\\
216	4\\
217	5\\
218	3\\
219	4\\
220	2\\
221	2\\
222	1\\
223	2\\
224	1\\
225	1\\
226	1\\
227	1\\
228	1\\
};
\end{axis}
\end{tikzpicture}%

%% file: eON6L1tex.tex
%
%
\definecolor{mycolor1}{rgb}{0.00000,0.44700,0.74100}%
\begin{tikzpicture}[scale=\tikzscale]

\begin{axis}[%
width=4.521in,
height=3.566in,
at={(0.758in,0.481in)},
scale only axis,
xmin=0,
xmax=70,
ymin=0,
ymax=120,
axis background/.style={fill=white},
title style={font=\bfseries},
title={Even orthogonal model. N= 6,L= 1}
]
\addplot [color=mycolor1, forget plot]
  table[row sep=crcr]{%
1	1\\
2	1\\
3	1\\
4	2\\
5	2\\
6	3\\
7	5\\
8	5\\
9	7\\
10	8\\
11	12\\
12	12\\
13	17\\
14	16\\
15	24\\
16	22\\
17	32\\
18	28\\
19	42\\
20	35\\
21	51\\
22	43\\
23	62\\
24	49\\
25	72\\
26	56\\
27	82\\
28	62\\
29	89\\
30	66\\
31	97\\
32	69\\
33	99\\
34	70\\
35	101\\
36	69\\
37	100\\
38	66\\
39	96\\
40	62\\
41	89\\
42	56\\
43	82\\
44	49\\
45	71\\
46	43\\
47	62\\
48	35\\
49	51\\
50	28\\
51	41\\
52	22\\
53	31\\
54	16\\
55	24\\
56	12\\
57	17\\
58	8\\
59	11\\
60	5\\
61	7\\
62	3\\
63	4\\
64	2\\
65	2\\
66	1\\
67	2\\
};
\end{axis}
\end{tikzpicture}%